\documentclass{article}
\usepackage[margin=36mm]{geometry}
\usepackage{graphicx} 
\usepackage{amsmath,mathtools,amssymb}
\usepackage{amsthm}
\usepackage{tikz}
\usepackage{tikz-cd}
\usepackage{hyperref}
\usepackage[capitalise]{cleveref}
\usepackage{comment}
\usepackage[textsize=scriptsize]{todonotes}
\usepackage{lineno}
\usetikzlibrary{fadings}
\usetikzlibrary{intersections}
\usepackage{nicematrix}

\DeclareMathOperator{\mrg}{merge}
\DeclareMathOperator{\gr}{gr}
\DeclareMathOperator{\Set}{Set}
\DeclareMathOperator{\img}{Im}

\newcommand{\AS}{\widehat{S}}
\newcommand{\R}{\mathbb{R}}
\newcommand{\Z}{\mathbb{Z}}
\newcommand{\F}{\mathbb{F}}
\newcommand{\Pp}{\mathcal{P}}
\newcommand{\Ss}{\mathcal{S}}
\newcommand{\Q}{\mathcal{Q}}
\newcommand{\M}{\mathcal{M}}
\newcommand{\sse}{\subseteq}

\newtheorem{theorem}{Theorem}
\newtheorem{lemma}[theorem]{Lemma}
\newtheorem{corollary}[theorem]{Corollary}
\newtheorem{definition}[theorem]{Definition}
\newtheorem{remark}[theorem]{Remark}
\newtheorem{example}[theorem]{Example}

\numberwithin{theorem}{section}

\title{Computing $p$-presentation distances is hard}

\title{Computing $p$-presentation distances is hard}
\author{H\aa vard Bakke Bjerkevik\thanks{Department of Mathematics \& Statistics, SUNY Albany, USA (current); \newline Department of Mathematics, Technical University of Munich, Germany (former).}\, and Magnus Bakke Botnan\thanks{Department of Mathematics, Vrije Universiteit Amsterdam, The Netherlands.}}

\begin{document}
\maketitle
\begin{abstract}
Recently, $p$-presentation distances for $p\in [1,\infty]$ were introduced for merge trees and multiparameter persistence modules as more sensitive variations of the respective interleaving distances ($p=\infty)$. It is well-known that computing the interleaving distance is NP-hard in both cases.
We extend this result by showing that computing the $p$-presentation distance is NP-hard for all $p\in [1,\infty)$ for both merge trees and $t$-parameter persistence modules for any $t\geq 2$.
Though the details differ, both proofs follow the same novel strategy, suggesting that our approach can be adapted to proving the NP-hardness of other distances based on sums or $p$-norms.
\end{abstract}


\section{Introduction}
Merge trees and persistence modules are related and fundamental concepts in topological data analysis (TDA). Starting with a function $f\colon X\to \R$, a merge tree 
tracks the evolution of the connected components of the sublevel sets of $f$. Such trees are useful in data visualization, have found applications in many settings, and can be readily computed. Similarly, persistence modules typically arise from applying homology to the nested family of sublevelsets. Perhaps the most classical setting is when $X=\mathbb{R}^n$, $Q\subset X$ is a finite subset of points, and $f(x) = \min_{q\in Q} ||x-q||_2$. These approaches to data science have found many applications to a wide range of disciplines \cite{dey2022computational,otter2017roadmap}. Recently, there has been an increasing interest in studying the homology of filtrations dependent on more than a single variable. The resulting objects are called multiparameter persistence modules, and how to effectively use them in data analysis is a topic of much ongoing research. We refer the reader to \cite{botnan2022introduction} for a recent overview.

In the literature, merge trees and (multiparameter) persistence modules are often compared using their respective interleaving distances. For single-parameter persistence, the interleaving distance is equal to the bottleneck distance, which is integral to many results central to TDA. However, despite its theoretical importance, the practitioner often prefers $p$-Wasserstein distances for a small $p$ \cite{carlsson2004persistence, cohen2010lipschitz,robinson2017hypothesis}; the case $p=\infty$ is precisely the bottleneck distance. This is because as $p$ decreases, the distances become relatively less sensitive to outlying intervals, measuring something closer to the sum of the differences between the modules rather than the maximal difference.
While the interleaving distance for merge trees \cite{de2016categorified} and multiparameter persistence modules \cite{lesnick2015theory} is an established part of the persistence literature, $p$-Wasserstein extensions only recently appeared \cite{bjerkevik2021ell, cardona2022universal}; another generalization that we shall not consider can be found in \cite{bubenik2023exact}. These distances are called $p$-presentation distances, and, importantly, it was also shown that these distances are all \emph{universal} in their respective settings. This naturally prompts the question of whether the $p$-presentation distances can be efficiently computed. For $p=\infty$, the computation was shown to be NP-hard for merge trees in \cite{agarwal2018computing} and for 2-parameter persistence modules in \cite{bjerkevik2020computing}.

\subsection{Contributions}
Our main two results are summarized in the following theorem.
\begin{theorem}
\label{thm:main-merge}
The problems of computing the $p$-presentation distance for merge trees and $t$-parameter persistence modules for any $t\geq 2$ are NP-hard for all $p\in [1, \infty)$. 
\end{theorem}
Our reductions are closely related to those for the case of $p=\infty$, but significantly more involved, requiring several novel ideas. Furthermore, our two proofs are the first hardness proofs for $\ell^p$ metrics in TDA that we are aware of, and we believe these proofs can serve as a blueprint for how to extend NP-hardness of distances involving maxima to NP-hardness of distances involving sums.
The following is an outline of the proof strategy in both cases:
\begin{itemize}
\item Associate to each instance of a NP-complete decision problem $C$ a pair of spaces (in our cases, merge trees or persistence modules) such that $C$ has a `yes' answer if and only if the two  spaces are close in the $\ell^\infty$-distance.
\item Show that if the two  spaces are close in the $\ell^\infty$-distance, then a small $\ell^p$-distance can be obtained by moving a set of objects (in our cases, generators or relations) a distance of at most $1$.
\item Show that independent of the instance of $C$, you have to move each object a distance of at least 1.
\item Show that if the $\ell^\infty$-distance is large (i.e., $C$ has a `no' answer), then there is at least one object that has to move a distance of more than 1, giving a larger $\ell^p$-distance than if the $\ell^\infty$-distance is small.
\end{itemize}
The fact that this proof strategy works for algebraic objects like persistence modules and combinatorial objects like merge trees suggests that our ideas can apply to a wide range of settings in TDA and beyond. While it is true that the distances we consider are closely related, the only essential element in the approach described above is that the distances are defined in terms of maxima and sums (or $p$-norms) of vectors describing movements of certain objects, and this is a natural basis for any definition of a sum-like distance, not just the presentation distances we consider.  

\subsection{Outline}
In \cref{sec:background}, we introduce basic notation, and recall $p$-Wasserstein distances. In \cref{sec:merge}, we prove the hardness result for merge trees, and in \cref{sec:multi]} we prove the analogous result for multiparameter persistence modules. 

\subsection{Acknowledgments}
The authors thank the Centre for Advanced Study (CAS) in Oslo for hosting the authors for a period in 2023 as part of the program Representation Theory: Combinatorial Aspects and Applications, where the main results of the paper were obtained.
The first author was funded by the Deutsche Forschungsgemeinschaft (DFG - German Research Foundation) - Project-ID 195170736 - TRR109.

\section{Background material}
\label{sec:background}
In this paper, we consider diagrams of vector spaces or sets indexed by a partially ordered set $P$. This means that we have objects $\{M_q\}_{q\in P}$ and  maps $\{M_{q\to q'}\colon M_q\to M_{q'}\}_{q\leq q'}$ satisfying the conditions that $M_{q\to q}$ is the identity on $M_q$ and $M_{r\to s}\circ M_{q\to r} = M_{q\to s}$ for all $q\leq r \leq s$. In the case of vector spaces, each $M_q$ is a vector space over some fixed field $\mathbb{F}$, and the morphisms are linear maps. For sets, $M_q$ is a set, and the maps are set maps. Considering $P$ as a category $P$ with objects the elements of $P$, and a unique morphism from $a$ to $b$ if $a\leq b$, the object above is a \emph{functor} $M\colon P\to \mathcal{C}$ where $\mathcal{C}$ is either the category $\mathrm{Vec}$ of vector spaces over $\mathbb{F}$, or the category $\mathrm{Set}$ of sets. We say that two diagrams $M$ and $N$ are \emph{isomorphic} if they are isomorphic as functors. Explicitly, this means that there is a family of isomorphisms $\{\phi_q\colon M_q\to N_q\}$ such that $\phi_r\circ M_{q\to r} = N_{q\to r}\circ\phi_q$ for all $q\leq r$. Here an isomorphism is either a linear isomorphism or a bijection, depending on which of the two settings we are working in. Although there is a growing interest in TDA in considering different posets $P$, we shall limit ourselves to $P = \R^n$.
\subsection{Persistence modules}

When $\mathcal{C} = \mathrm{Vec}$, we further say that $M\colon P\to \mathcal{C}$ is a \emph{persistence module}. Such modules typically arise by applying homology with field coefficients to the sublevel sets of a multivalued function. That is, if $f\colon X\to \mathbb{R}^n$, then we obtain a persistence module $N\colon \mathbb{R}^n\to \mathrm{Vec}$ by letting
\[ N_t = H_i\left(f^{-1}\{x\in X \mid f(x) \leq t\}; \mathbb{F}\right).\]

As is well known, if $P=\R$, and $\dim M_a < \infty$ for all $a$, then $M$ is uniquely described by a multiset $B(M)$ of intervals in $\R$ called the barcode of $M$ \cite{crawley2015decomposition}. In this work, we shall work under an even stronger assumption, namely that the persistence modules are \emph{finitely presented} (see \cref{subsec.mod.ppresdist}). For $n=1$, this is equivalent to saying that $B(M)$ contains a finite number of intervals and that each interval is of the form $[a,b)$ for $a<b$ and where $b$ is allowed to be infinite. Hence, an interval $[a,b)$ in the barcode can be identified with the point $(a,b)\in (\R\cup \{\infty\})^2$. The following definition of Wasserstein distances thus applies to barcodes of finitely presented persistence modules. 

\subsubsection{Wasserstein distances} Let $X$ and $Y$ be two finite subsets of $\{ (a,b) \mid a<b\}\subset (\R\cup \{\infty\})^2$, and let $\Pi\colon \R^2\to \{(a,a) \mid a\in \R\}$ be the orthogonal projection onto the diagonal in $\R ^2$. A bijection $\phi\colon X\supseteq A\to B\subseteq Y$ is called a \emph{matching} between $X$ and $Y$. For any such $\phi$ and $p\in [0, \infty)$, we let 
\[ \text{$p$-cost}(\phi) = \left(\sum_{x\in A}||x-\phi(x)||_p^p + \sum_{z\in (X\setminus A)\cup (Y\setminus B)} ||z-\Pi(z)||_p^p\right)^{1/p},
\]
where follow the convention that
\[||(a, \infty) - (a',\infty)||_p  = |a-a'| \qquad \text{and} \qquad ||(a, \infty) - \Pi((a, \infty))||_p = \infty.\]

For $p=\infty$, 
\[
\text{$p$-cost}(\phi) = \max\left(\max_{x\in A}||x-\phi(x)||_p^p, \max_{z\in (X\setminus A)\cup (Y\setminus B)} ||z-\Pi(z)||_p^p\right).
\]
The \emph{$p$-Wasserstein distance} \cite{carlsson2004persistence,cohen2010lipschitz,robinson2017hypothesis} between $X$ and $Y$ is \[d^p_W(X,Y) = \min_{\substack{\text{matchings }\\ \phi \colon X\to Y}} \text{$p$-cost}(\phi).\]

\subsection{Merge trees}
\label{sec:basicmerge}
A merge tree is classically defined through the sublevel sets of a sufficiently tame function $f\colon X\to \R$. Examples include when $f$ is a Morse function on a compact manifold, or a piecewise-linear function on a finite simplicial complex. The \emph{epigraph} of $f$ is the  set $E_f = \{(x,t) \mid f(x) \leq t\}$, and the (geometric) merge tree $\mathcal{M}_f$ of $f$ is the \emph{Reeb graph} of the projection map $\Pi\colon E_f\to \R$ defined by $\pi(x,t) = t$. Explicitly, 
\[ \mathcal{M}_f = E_f /{\sim} \quad \text{where} \quad (x,s) \sim (y,t) \iff s=t \text{ and } [x] = [y] \in \pi_0(\Pi^{-1}(t)). \]
Here $\pi_0$ denotes the functor that takes a topological space to its set of connected components. Since equivalent points in the quotient space have the same function value by definition, $\Pi$ descends to a continuous map $\tilde \Pi\colon \mathcal{M}_f\to \R$. Both $(X,f)$ and $(\mathcal{M}_f, \tilde{\Pi})$ are examples of $\R$-spaces\footnote{An $\R$-space is a pair $(X,f)$, where $f\colon X\to \R$ is continuous. Two $\R$-spaces $(X,f)$ and $(Y,g)$ are isomorphic if there exists a homeomorphism $h\colon X\to Y$ such that $g\circ h = f$.}. An illustration of a merge tree is given in \cref{fig.merge}. 

Associated to $f$, we also obtain a functor $M_f\colon \R\to \mathrm{Set}$ by letting $(M_f)_t = \pi_0(f^{-1}(-\infty, t])$, and observing that we have a set map $(M_f)_{s\to t}\colon (M_f)_s\to (M_f)_t$ induced by the inclusion $f^{-1}(-\infty, s]\hookrightarrow f^{-1}(-\infty, t]$. The salient point is that the functor $M_f$ and the geometric merge tree $\mathcal{M}_{\tilde{\Pi}}$ determine each other up to isomorphism. Firstly, since $\mathcal{M}_{\tilde{\Pi}}$ was constructed precisely to capture the connected components of sublevel sets, we have natural bijections $\pi_0(f^{-1}(-\infty, t]) \cong \pi_0(\tilde{\Pi}^{-1}(-\infty,t])$ and therefore $M_{f}\cong M_{\tilde{\Pi}}$. Conversely, from $M_f$ one can construct a geometric merge tree $|M_f|$, called the \emph{geometric realization}, which is isomorphic to $(\mathcal{M}_f, \tilde{\Pi})$ as an $\R$-space. A formal approach to the interplay between functors from $\R$ to $\Set$ and geometric merge trees can be found in \cite{cardona2022universal}.

\begin{figure}
\centering
\begin{tikzpicture}

    \draw[thick, name path=upper] plot[smooth, tension=1] coordinates {(-4,0) (-3,2) (-2,-1) (-1,1) (0,-2) (1,0.5)};

    \fill[gray!50, opacity=0.5] 
        (-4,3) -- (1,3) -- 
        plot[smooth, tension=1] coordinates {(1,0.5) (0,-2) (-1,1) (-2,-1) (-3,2) (-4,0)} -- cycle;
        
    \draw[dashed, black!50] (-4,0) -- (7,0);
       \draw[dashed, black!50] (-3,2) -- (7,2);
      \draw[dashed, black!50](-1,1) -- (7,1);
       \draw[dashed, black!50](0,-2) -- (7,-2);
          \draw[dashed, black!50] (-2,-1) -- (7,-1);

          \node[right] at (7,-2) {$c_1$};
           \node[right] at (7,-1) {$c_2$};
            \node[right] at (7,0) {$c_3$};
             \node[right] at (7,1) {$c_4$};
              \node[right] at (7,2) {$c_5$};

              \node at (0.5,2.5) {$E_f$};
              \node at (-3.5,1) {$f$};
                 
    \begin{scope}[xshift=6cm]
    \draw[thick] plot[smooth, tension=0] coordinates {(-4,0) (-3,2) (-1,1) (0,-2)};
    \draw[thick] (-2,-1) to (-1,1);
    
    \foreach \x/\y in {-4/0, -3/2, -2/-1, -1/1, 0/-2} {
        \fill (\x,\y) circle (3pt);
    }

    \draw[->, thick](-3,2) -- (-3,3);
        \node at (-2.5,2.5) {$\mathcal{M}_f$};

    \node[below] at (0,-2.1) {$a$};
     \node[below] at (-4,-0.1) {$c$};
    \node[below] at (-2,-1.1) {$b$};

    \end{scope}

    \begin{scope}[xshift=12cm]

    \foreach \x/\y in {-4/-2, -3.5/-1, -3/0} {
        \fill (\x,\y) circle (3pt);
    }
        \foreach \x/\y in {-3.5/1, -3/2} {
        \draw (\x,\y) circle (3pt);
    }

    \draw[thick,->] (-4,-2) -- (-4, 3);
    \draw[thick] (-3.5,-1) -- (-3.5, 1);
    \draw[thick] (-3,0) -- (-3, 2);
        \node at (-3,2.5) {$B(\mathcal M_f)$};
    \end{scope}
\end{tikzpicture}
\caption{A function $f\colon X\to \mathbb{R}$ together with its epigraph $E_f$ on the left, the associated (geometric) merge tree $\mathcal{M}_f$ in the middle (where the local minima have been labeled), and the $H_0$ barcode of sublevel filtration of $\mathcal{M}_f$ to the right. We observe that the merge tree contains precise information about which branches merge, whereas this information is lost in the barcode.}
\label{fig.merge}
\end{figure}

\subsubsection{Merge trees through presentations}
\label{sec:basic}

To introduce the $p$-presentation distance for merge trees, we must first discuss presentations of merge trees. In fact, we shall use this to give a purely categorial definition of a merge tree without reference to a real-valued function on a topological space. To avoid extensive use of categorical language, we work with sets and relations. However, our constructions are equivalent to those given in \cite{cardona2022universal}. 

A \emph{merge tree presentation} (referred to as just a \emph{presentation} in this section) is a triplet $P = (G,R, \gr)$ where $G$ is a finite set, $R$ is a multiset of subsets of $G$ of cardinality 2, and $\gr$ is a function from $G\cup R$ to $\R$ such that for all $r = \{g,g'\}\in R$, $\gr(r)\geq \max\{\gr(g),\gr(g')\}$.
We refer to $G$ as the set of \emph{generators}, $R$ as the set of \emph{relations}, and $\gr$ as the \emph{grading function}.
For $s\in \R$, let $G_s = \{ g\in G \mid \gr(g) \leq s\}$ and $R_s = \{ r\in R \mid \gr(r) \leq s\}$,  and let 
$\sim_s$ denote the equivalence relation on $G_s$ generated by $g\sim_s h$ if  $\{g,h\}\in R_s$. Associated to any such triplet there  is a functor $F(G,R,\gr)\colon \R\to \mathrm{Set}$ given by 
\[ F(G,R,\gr)_s = G_s/{\sim}_s \]
and where the maps are induced by the inclusions of sets $G_s\subseteq G_t$ for $s\leq t$.

\begin{definition}
    A functor $M\colon \R \to \mathrm{Set}$ is a \emph{merge tree} if there exists a triplet $(G,R, \gr)$ such that $M\cong F(G,R,\gr)$ and the cardinality of $M_s$ is 1 for all sufficiently large $s$. We shall refer to such a triplet as a \emph{presentation} of $M$.
\end{definition}
\begin{remark}
    The condition on the cardinality for large $s$ is to ensure that all branches ultimately merge; without this condition we would have a merge \emph{forest}.
\end{remark}

This definition of a merge tree is slightly less general than the one given in \cite{cardona2022universal} in that we assume the generator and relation sets to be finite.
However, this finiteness condition is implicitly assumed in all the results of \cite{cardona2022universal}, so we are effectively working at the same level of generality.

Note that the functors $M_f$ obtained in the previous section are all merge trees according to this definition provided $f$ is sufficiently tame. 
\begin{example}
\label{ex.mergetree1}
Consider the geometric merge tree $\mathcal{M}_f$ in the middle of \cref{fig.merge}. This merge tree is a geometric realization of $M_f$, and for $t_1 < t_2$ in an interval $(c_i, c_{i+1})$, the merge tree remains constant. It therefore suffices to specify the functor $M_f\colon \R\to \mathrm{Set}$ at the critical values $c_1, \ldots, c_5$. Using the first-appearing vertex in each component as a representative of the connected component,
\[\{[a]\} \xrightarrow{[a]\mapsto [a]} \{[a],[b]\}\xrightarrow{\substack{[a]\mapsto [a] \\ [b]\mapsto [b]}}\{[a], [b], [c]\}\xrightarrow{\substack{[a]\mapsto [a] \\ [b]\mapsto [a] \\ [c]\mapsto [c]}} \{[a], [c]\}\xrightarrow{\substack{[a]\mapsto [a] \\ [c]\mapsto [a]}} \{[a]\}. \]
We obtain a presentation $P$ by letting $G$ be the leaves, $G=\{[a], [b], [c]\}$, and the relations $R=\{\{[a], [b]\}, \{[a], [c]\}\}$ and 
\[\gr([a]) = c_1 \qquad \gr([b])= c_2 \qquad \gr([c]) = c_3 \qquad \gr\{[a], [b]\} = c_4 \qquad \gr\{[a], [b]\} = c_5.  \]
We can visualize the generators as strands, and the relations as merging strands, from which how to recover the geometric merge tree becomes clear.
\begin{center}
\begin{tikzpicture}
        \draw[gray!40, dashed] (1,1) grid (8,-2);
        \draw[->, thick, blue] (1,1) -- (8,1) node[right] {$[a]$};
        \draw[->, thick, blue] (2,0) -- (8,0) node[right] {$[b]$};;
        \draw[->,thick, blue] (3,-1) -- (8,-1) node[right] {$[c]$};;

        \fill[red] (4,-2) circle (3pt) node[left] {$\{[a], [b]$\}};
        \fill[red] (5,-2) circle (3pt) node[right] {$\{[a], [c]$\}};

        \draw[red, thick, bend right] (4,-2) to (4,1);
        \draw[red, thick, bend left] (4,-2) to (4,0);
         \draw[red, thick, bend right] (5,-2) to (5,1);
        \draw[red, thick, bend left] (5,-2) to (5,-1);
        
\end{tikzpicture}
\end{center}
Note that this presentation is non-unique. Firstly, replacing the relation $\{[a], [c]\}$ with $\{[b], [c]\}$ gives another presentation for $M_f$. One can also add an additional generator $g$ with $\gr(g)\geq c_1$, together with the relation $r=\{[a], [g]\}$ with $\gr(r) = \gr(g)$. In conclusion, there is quite a lot of freedom in choosing a presentation, and this will be important when discussing compatible presentations below. 
\end{example}

When a presentation $(G,R,\gr)$ of $M$ is clear from the context, we shall use $[g]\in M_{s}$, for $s\geq \gr([g])$, to denote the equivalence class of $g\in G_s$. For $a\geq \gr(g)$ we define $M_a(g)$ to be the equivalence class of $[g]$ under the map $G_{\gr(g)}\to G_a$.  We shall employ the notation $g=h$ for a relation $\{g,h\}$.

\subsubsection{Barcodes of merge trees}
By considering each $M_t$ as a topological space with the discrete topology, we can apply the $0$-th homology functor $H_0(-;\mathbb{F})$ to get a persistence module $H_0(M; \mathbb{F})$ defined pointwise by $H_0(M; \mathbb{F})_t = H_0(M_t; \mathbb{F})$. For any merge tree $M$, we thus define  $B(M) := B(H_0(M;\mathbb{F}))$. Unsurprisingly, $B(M)$ coincides with the `standard' $0$-th persistent homology of the sublevel set filtration of the geometric realization of $M$ \cite{cardona2022universal}. 

\begin{example}[\cref{ex.mergetree1} continued.]
    Applying $H_0$\footnote{This is the same as applying the free functor from sets to vector spaces, i.e., linearization of the set maps.} to\[\{[a]\} \xrightarrow{[a]\mapsto [a]} \{[a],[b]\}\xrightarrow{\substack{[a]\mapsto [a] \\ [b]\mapsto [b]}}\{[a], [b], [c]\}\xrightarrow{\substack{[a]\mapsto [a] \\ [b]\mapsto [a] \\ [c]\mapsto [c]}} \{[a], [c]\}\xrightarrow{\substack{[a]\mapsto [a] \\ [c]\mapsto [a]}} \{[a]\}, \]
    we obtain
    \[\F \xrightarrow{{\begin{bmatrix} 1 \\ 0 \end{bmatrix}}} \F^2\xrightarrow{\begin{bmatrix} 1 & 0 \\ 0 & 1 \\ 0 & 0 \end{bmatrix}}\F^3\xrightarrow{\begin{bmatrix} 1 & 1 & 0 \\ 0 & 0 & 1\end{bmatrix}}\F^2\xrightarrow{\begin{bmatrix} 1 & 1 \end{bmatrix}} \F. \]  
The barcode of this persistence module can be obtained by changing bases or by applying the inclusion-exclusion formula to compute the barcode (persistence diagram) from the ranks of the linear maps; see, e.g., \cite{cohen2005stability}. This yields the barcode  
\[
\{ [c_1, \infty), [c_2, c_4), [c_3, c_5) \}
\]
exactly as for the geometric merge tree in \cref{fig.merge}.
\end{example}

\section{Computing $d_I^p$ for Merge trees}
\label{sec:merge}
In this section we show that computing the $p$-presentation distance between merge trees is NP-hard for all $p\in [0,\infty)$; NP-hardness for $p=\infty$ follows from the proof of \cite[Theorem 3.3]{agarwal2018computing}, and our argument here is an adaptation of the reduction found in that work.

\subsection{The $p$-presentation distance}Two presentations $P= 
(G_P, R_P, \gr_P)$ and $Q=(G_Q, R_Q, \gr_Q)$ of not necessarily isomorphic merge trees $M$ and $N$ are $\sigma$-\emph{compatible} if there is a bijection $\sigma\colon G_P\cup R_P\to G_Q\cup R_Q$ restricting to bijections $G_P\to G_Q$ and $R_P\to R_Q$ such that $\sigma(\{g,h\}) = \{\sigma(g), \sigma(h)\}$ for all $\{g,h\}\in R_P$. If $\sigma$ can be chosen to be the identify function, then we say that the two presentations are \emph{compatible}. It is shown in \cite[Lemma~3.8]{cardona2022universal} that any two merge trees admit compatible presentations. We shall assume throughout this section that $\sigma$ is the identity.

For compatible presentations $P$ and $Q$, and $p\in [1, \infty]$, we define
\[
d^p(P,Q) = \begin{cases} \left(\sum_{h\in G_P\cup R_P} |\gr_P(h)-\gr_Q(h)|^p\right)^{\frac 1 p} & \text{if $p<\infty$} \\
\max_{h\in G_P\cup R_P} \left|\gr_P(h)-\gr_Q(h)\right| & \text{if $p=\infty$.}\end{cases}
\]
For merge trees $M$ and $N$, 
\[\hat d_I^p(M,N) = \inf\{d^p(P,Q) \mid P\text{ and }Q \text{ are compatible presentations of }M\text{ and } N\}.\]

\begin{example}
\label{ex.comppres}
Returning to the merge tree $M_f$ from \cref{ex.mergetree1} with presentation $P$, let us also consider another merge tree $N$ with a presentation $Q$ with generators $\{g_1, g_2\}$ with $\gr(g_1) = \gr(g_2) = c_1$ and a relation $r = \{g_1,g_2\}$ with $\gr(r) = c_2$. Labeling the rows and columns by the grades of the generators and relations, respectively, we can represent the presentations by matrices as follows,
\begin{equation}
P=
\begin{bNiceMatrix}[
  first-row,code-for-first-row=\scriptstyle,
  first-col,code-for-first-col=\scriptstyle,
]
& c_4 & c_5 \\
c_1 & 1 & 1 \\
c_2 & 1 & 0\\
c_3 & 0 & 1 
\end{bNiceMatrix}\qquad \qquad Q = \begin{bNiceMatrix}[
  first-row,code-for-first-row=\scriptstyle,
  first-col,code-for-first-col=\scriptstyle,
]
& c_2  \\
c_1 & 1 \\
c_1 & 1
\end{bNiceMatrix}
\end{equation}

Discarding the labels, two presentations are compatible (with $\sigma$ as the identity) if and only if their matrices are identical. A straightforward way to achieve identical matrices is to take their direct sum, ensuring that the added rows and columns have the same grade as the first appearing generator in each respective presentation. Finally, an additional relation is introduced to merge the newly added generators with the first appearing generator in the original presentation. 
Specifically for the example above, we get 
\begin{equation}
\hat{P}=
\begin{bNiceMatrix}[
  first-row,code-for-first-row=\scriptstyle,
  first-col,code-for-first-col=\scriptstyle,
]
& c_4 & c_5 & {\color{red} c_1} &{\color{red} c_1} \\
c_1 & 1 & 1 & {\color{red} 0}  & {\color{red} 1}  \\
c_2 & 1 & 0 & {\color{red} 0}  & {\color{red} 0}  \\
c_3 & 0 & 1 & {\color{red} 0}  & {\color{red} 0}   \\
{\color{red} c_1} & {\color{red} 0}  & {\color{red} 0} & {\color{red} 1} & {\color{red} 1} \\
{\color{red} c_1} & {\color{red} 0}  & {\color{red} 0}  & {\color{red} 1} & {\color{red} 0} 
\end{bNiceMatrix}\qquad \qquad \hat{Q} = \begin{bNiceMatrix}[
  first-row,code-for-first-row=\scriptstyle,
  first-col,code-for-first-col=\scriptstyle,
]
& {\color{red} c_1} & {\color{red} c_1} & c_2 & {\color{red} c_1} \\
{\color{red} c_1}& {\color{red} 1} & {\color{red} 1}  & {\color{red} 0}  & {\color{red} 1}  \\
{\color{red} c_1} & {\color{red} 1}  & {\color{red} 0}  & {\color{red} 0}  & {\color{red} 0}  \\
{\color{red} c_1} & {\color{red} 0}  & {\color{red} 1}  & {\color{red} 0}  & {\color{red} 0}   \\
c_1 & {\color{red} 0}  & {\color{red} 0}  & 1 & {\color{red} 1}  \\
c_1 & {\color{red} 0}  & {\color{red} 0}  & 1 & {\color{red} 0} 
\end{bNiceMatrix}.
\end{equation}
In the following, we shall assume that the critical values are equidistant. One then sees that, $d^\infty(\hat{P}, \hat{Q}) = c_5-c_1$. This can be further improved, but it is not necessarily true that matrices of minimal dimensions minimize $d^\infty$. Consider for instance the compatible presentations, 
\begin{equation}
P=
\begin{bNiceMatrix}[
  first-row,code-for-first-row=\scriptstyle,
  first-col,code-for-first-col=\scriptstyle,
]
& c_4 & c_5 \\
c_1 & 1 & 1 \\
c_2 & 1 & 0\\
c_3 & 0 & 1 
\end{bNiceMatrix}\qquad \qquad \hat{Q}' = \begin{bNiceMatrix}[
  first-row,code-for-first-row=\scriptstyle,
  first-col,code-for-first-col=\scriptstyle,
]
& c_2 & c_4 \\
c_1 & 1 & 1 \\
c_1 & 1 & 0\\
c_4 & 0 & 1 
\end{bNiceMatrix}.
\end{equation}
For this case, we get $d^\infty (P, \hat{Q}') = c_4-c_2$. To improve on this, consider
\begin{equation}
\bar{P} =
\begin{bNiceMatrix}[
  first-row,code-for-first-row=\scriptstyle,
  first-col,code-for-first-col=\scriptstyle,
]
& c_4 & c_5 & (c_1+c_2)/2 \\
c_1 & 1 & 1 & 1 \\
c_2 & 1 & 0 & 0 \\
c_3 & 0 & 1 & 0 \\
(c_1+c_2)/2 & 0 & 0 & 1 
\end{bNiceMatrix}\qquad \qquad \bar{Q} = \begin{bNiceMatrix}[
  first-row,code-for-first-row=\scriptstyle,
  first-col,code-for-first-col=\scriptstyle,
]
& c_3 & c_4 & c_2 \\
c_1 & 1 & 1 & 1 \\
c_3 & 1 & 0 & 0 \\
c_4 & 0 & 1 & 0 \\
c_1 & 0 & 0 & 1 
\end{bNiceMatrix}.
\end{equation}
from which we obtain 
\[\hat d^\infty(\bar{P}, \bar{Q}) =  c_4 - c_3.\]
This realizes $\hat d_I^\infty(M_f, N)$ since $d^\infty(\bar{P}, \bar{Q})\geq \hat d_I^\infty(M_f,N) = d_I(M_f, N)$ and $d_I(M_f,N)$ is known to be lower-bounded by $d_B(B(M_f), B(N)) = c_4-c_3$ \cite{morozov2013interleaving}; the optimal matching is given by matching the infinite bars and leaving other bars unmatched. The barcodes are shown in \cref{fig.merge2}.
\end{example}
\begin{figure}
\centering
\begin{tikzpicture}

\draw[thick] plot[smooth, tension=0] coordinates {(-4,0) (-3,2) (-1,1) (0,-2)};
    \draw[thick] (-2,-1) to (-1,1);
    
    \foreach \x/\y in {-4/0, -3/2, -2/-1, -1/1, 0/-2} {
        \fill (\x,\y) circle (3pt);
    }

    \draw[->, thick](-3,2) -- (-3,3);
        \node at (-2.5,2.5) {$|M_f|$};

    \node[below] at (0,-2.1) {$a$};
     \node[below] at (-4,-0.1) {$c$};
    \node[below] at (-2,-1.1) {$b$};
        
    \draw[dashed, black!50] (-4,0) -- (5,0);
       \draw[dashed, black!50] (-3,2) -- (5,2);
      \draw[dashed, black!50](-1,1) -- (5,1);
       \draw[dashed, black!50](0,-2) -- (5,-2);
          \draw[dashed, black!50] (-2,-1) -- (5,-1);

          \node[right] at (5,-2) {$c_1$};
           \node[right] at (5,-1) {$c_2$};
            \node[right] at (5,0) {$c_3$};
             \node[right] at (5,1) {$c_4$};
              \node[right] at (5,2) {$c_5$};

    \begin{scope}[xshift=6cm]
    \draw[thick] plot[smooth, tension=0] coordinates { (-4,-2)  (-3,-1) (-2,-2) };

    \foreach \x/\y in {-4/-2, -3/-1, -2/-2} {
        \fill (\x,\y) circle (3pt);
    }

    \draw[->, thick](-3,-1) -- (-3,3);
        \node at (-2.5,2.5) {$|N|$};

    \node[below] at (-2,-2.1) {$g_1$};
    \node[below] at (-4,-2.1) {$g_2$};

    \end{scope}

    \begin{scope}[xshift=10cm]

    \foreach \x/\y in {-4/-2, -3.5/-1, -3/0, -1.5/-2, -1/-2} {
        \fill (\x,\y) circle (3pt);
    }
        \foreach \x/\y in {-3.5/1, -3/2, -1/-1} {
        \draw (\x,\y) circle (3pt);
    }

    \draw[thick,->] (-4,-2) -- (-4, 3);
        \draw[thick,->] (-1.5,-2) -- (-1.5, 3);
        \draw[<->,red,thick,dashed] (-3.8,-2)--(-1.7,-2);
         \draw[thick] (-1,-2) -- (-1, -1);
    \draw[thick] (-3.5,-1) -- (-3.5, 1);
    \draw[thick] (-3,0) -- (-3, 2);
        \node at (-3,2.5) {$B(M_f)$};
                \node at (-0.8,2.5) {$B(N)$};
    \end{scope}
\end{tikzpicture}
\caption{Geometric realizations of the merge trees $M_f$ and $N$ from \cref{ex.comppres} and their barcodes. The dashed lines represents the optimal matching in the computation of $d_W^p(B(M_f),B(N))$ for all $p$.}
\label{fig.merge2}
\end{figure}

For $p\neq\infty$, $\hat d_I^p$ does not satisfy the triangle inequality.
The following example is a rescaling of \cite[Example 4.2]{cardona2022universal} showing this.
\begin{example}
\label{ex_triangle_fail}
Consider the three merge trees $M_1$, $M_2$ and $M_3$ in \cref{fig.no_triangle_ineq}.
Below, a presentation $P_1$ of $M_1$ is given on the left, presentations $P_2$ and $P'_2$ of $M_2$ are in the middle, and a presentation $P_3$ of $M_3$ is on the right,
\[
\begin{tikzcd}[row sep=0.04cm,column sep=2cm, ampersand replacement=\&]
    P_1 \& P_2 \& P_2' \& P_3 \\
    \begin{bNiceMatrix}[
  first-row,code-for-first-row=\scriptstyle,
  first-col,code-for-first-col=\scriptstyle,
]
& \epsilon \\
0 & 1 \\
0 & 1 \\
\end{bNiceMatrix} \ar[r, leftrightarrow, "compatible"] \&  \begin{bNiceMatrix}[
  first-row,code-for-first-row=\scriptstyle,
  first-col,code-for-first-col=\scriptstyle,
]
& 0 \\
0 & 1 \\
0 & 1 \\
\end{bNiceMatrix}
\& \begin{bNiceMatrix}[
  first-row,code-for-first-row=\scriptstyle,
  first-col,code-for-first-col=\scriptstyle,
]
 \\
0 \\
\end{bNiceMatrix}
\ar[r, leftrightarrow, "compatible"] 
\&
\begin{bNiceMatrix}[
  first-row,code-for-first-row=\scriptstyle,
  first-col,code-for-first-col=\scriptstyle,
]
 \\
1 \\
\end{bNiceMatrix}
\end{tikzcd} 
\]
The presentations $P_1$ and $P_2$ are compatible, and therefore $\hat d^p_I(M_1,M_2)\leq \epsilon$, which we can assume to be smaller than $2^{\frac 1p}-1$. Likewise, $P'_2$ and $P_3$ are also compatible and, therefore, $\hat d^p_I(M_2,M_3) \leq 1$.
But by an argument similar to the one given in \cite{cardona2022universal}, $\hat d^p_I(M_1,M_3) \geq 2^{\frac 1p}$; intuitively, this is because we need to move two generators from $0$ to $1$.
It follows that the triangle inequality does not hold for $\hat d^p_I$, because we can sometimes take ``shortcuts'' through two different presentations of an intermediate module.
\end{example}

\begin{figure}
\centering
\begin{tikzpicture}

       \draw[dashed, black!50] (6,2) -- (9,2);
      \draw[dashed, black!50] (0,-1) -- (9,-1);
       \draw[dashed, black!50](-1,-2) -- (9,-2);

          \node[right] at (9,-2) {$0$};
           \node[right] at (9,-1) {$\epsilon$};
              \node[right] at (9,2) {$1$};
              
    \begin{scope}[xshift=0cm]
    \draw[thick] plot[smooth, tension=0] coordinates {(-1,-2) (0,-1) (1,-2)};
    \draw[thick,->] (0,-1) to (0,3);
    \node at (.5,2.5) {$M_1$};
    
    \foreach \x/\y in {-1/-2, 0/-1, 1/-2} {
        \fill (\x,\y) circle (3pt);
    }
    \end{scope}
                 
    \begin{scope}[xshift=3cm]
    \draw[thick,->] (0,-2) to (0,3);
    \node at (.5,2.5) {$M_2$};
    
    \foreach \x/\y in {0/-2} {
        \fill (\x,\y) circle (3pt);
    }
    \end{scope}
                 
    \begin{scope}[xshift=6cm]
    \draw[thick,->] (0,2) to (0,3);
    \node at (.5,2.5) {$M_3$};
    
    \foreach \x/\y in {0/2} {
        \fill (\x,\y) circle (3pt);
    }
    \end{scope}
\end{tikzpicture}
\caption{An example borrowed from \cite{cardona2022universal} where for a sufficiently small $\epsilon$ depending on $p$, which is assumed to be finite, $\hat d^p_I(M_1,M_3)>\hat d^p_I(M_1,M_2)+\hat d^p_I(M_2,M_3)$.}
\label{fig.no_triangle_ineq}
\end{figure}
This observation prompts the following definition. 

\begin{definition}[\cite{cardona2022universal}]
\label{def.ppres.merge}
For $p\in [1,\infty]$, the $p$-presentation distance between two merge trees $M$ and $N$ is
\[
d_I^p(M,N) = \inf\sum_{i=1}^{\ell-1} \hat{d}_I^p(M_i,M_{i+1}),
\]
where the infimum is taken over all sequences $M=M_1, M_2,\dots,M_\ell=N$ of merge trees. 
\end{definition}

For $p=\infty$, $d_I^p$ is the well-known \emph{interleaving distance} for merge trees \cite[Theorem~1.4]{cardona2022universal}.
As illustrated in \cref{ex_triangle_fail}, using intermediate modules is sometimes needed to find $d_I^p(M,N)$:
$\hat{d}_I^p(M_1,M_3) > \hat{d}_I^p(M_1,M_2) + \hat{d}_I^p(M_2,M_3)$, so the infimum in \cref{def.ppres.merge} is not realized by $\hat{d}_I^p(M_1,M_3)$.
This becomes relevant when we need to lower bound $d^p_I(M,N)$ for two merge trees $M$ and $N$ in \cref{sec_merge_l_bound}; much of the difficulty in obtaining the lower bound comes from the fact that we need to consider all sequences $M=M_1, M_2,\dots,M_\ell=N$ of merge trees, instead of just looking at compatible presentations of $M$ and $N$.

\begin{example}[Continuing \cref{ex.comppres}]
    For any finite $p$, 
    \begin{align*}
    d^p(\bar{P}, \bar{Q})^p &= (c_4-c_3)^p + (c_5-c_4)^p + 2((c_2-c_1)/2)^p + (c_3-c_2)^p + (c_4-c_3)^p\\
    &= 4(c_2-c_1)^p + 2^{1-p}(c_2-c_1)^p = (4+2^{1-p})(c_2-c_1)^p.
    \end{align*}
The $p$-Wasserstein distance is obtained (as for the case $p=\infty)$ by matching the two long bars, and leaving the rest unmatched. A small computation, using that $(c_4-c_2) = (c_5-c_3) = 2(c_2-c_1)$, gives that 
\[d_W^p(B(M_f), B(N))^p = (4+2^{1-p})(c_2-c_1)^p.\]
By \cite[Theorem 1.2]{cardona2022universal},
\[d_W^p(B(M_f), B(N)) \leq d_I^p(M_f, N) \leq d^p(\bar{P}, \bar{Q}),\]
and therefore $d_I^p(M_f, N)$ is realized by the presentation matrices $\bar{P}$ and $\bar{Q}$ without the need for a sequence of merge trees as in \cref{def.ppres.merge}. 
\end{example}

\subsection{Merge trees and the partition problem}
\label{sec:3part}
Let $S=\{s_1, \ldots, s_m\}$ be a multiset  of positive integers, $n=s_1 + \ldots + s_m\geq 2$, and let $1\leq k\leq m$ be an integer. The computational problem \emph{Balanced partition} (BAL-PART) asks if there exists a partition $S = \cup_{i=1}^k S_i$ such that, $S_i\cap S_j = \emptyset$ for $i\neq j$, and 
\[\AS \coloneqq \frac{n}{k} = \sum_{s\in S_i} s, \qquad \text{for all $i$}.\]
It may be assumed that $n$ is a multiple of $k$. This problem is known to be strongly NP-complete by a reduction from 3-Partition \cite{agarwal2018computing}. 

Throughout this section, we will consider a fixed instance $\Ss \coloneqq (S,k)$ of BAL-PART, as defined above. Associated to $\Ss$ we construct the following merge trees $M$ and $N$. In the following, $C=4n^{\frac 1p}$.

\subsubsection{The merge tree $M$.} Let $M$ be the merge tree with $n$ generators
\[g_{1,1},\dots,g_{1,s_1},g_{2,1},\dots,g_{2,s_2},\dots,g_{m,1},\dots,g_{m,s_m}\]
at $-C$, together with relations $r_{i,j}: g_{i,j} = g_{i,j+1}$ at grade $0$ for all $1\leq i\leq m$ and $1\leq j\leq s_i-1$. In addition, there are relations $r_i: g_{i,s_i} = g_{i+1,1}$ at grade $2$ for $1\leq i\leq m-1$. We shall denote by $a_i$ the number of relations at grade $i$ in the presentation of $M$. The merge tree $M$ is illustrated in \cref{fig:mergeM}.

\subsubsection{The merge tree $N$.} Let $N$ be the merge tree with $n$ generators
\[h_{1,1},\dots,h_{1,\AS},h_{2,1},\dots,h_{2,\AS},\dots,h_{k,1},\dots,h_{k, \AS}\]
at $-C$, together with relations $r'_{i,j}: h_{i,j} = h_{i,j+1}$ at grade $1$ for all $1\leq i\leq k$ and $1\leq j\leq \AS-1$. In addition, there are relations $r'_i: h_{i,\AS} = h_{i+1,1}$ at grade $3$ for $1\leq i\leq k-1$. This is illustrated in \cref{fig:mergeN}. Let $b_1$ be the number of relations at $1$, and $b_3=n-1-b_1$ the number of relations at $3$. 

\subsubsection{Properties of $M$ and $N$}
Observe that $a_0\leq b_1$. Indeed, $a_0 = s_1+ \ldots + s_m - m = n-m$, $b_1 = (\hat{S}-1)\cdot k = n-k$, and $1\leq k\leq m$ by assumption. Since the total number of relations are equal in the two presentations, we have that $a_0+a_2 = b_1+b_3$. It thus follows that $a_2\geq b_3$. 

These observations are used in the following lemma.

\begin{lemma}
Let $p\in [1,\infty)$. Then, $d_W^p(B(M), B(N)) = (n-1)^\frac{1}{p}.$
\label{lem:dwp}
\end{lemma}
\begin{proof}
Both barcodes consist of $n$ intervals: $B(M) = \{[-C, \infty)_1, [-C, 0)_{a_0}, [-C, 2)_{a_2}\}$, and $B(N) = \{[-C, \infty)_1, [-C, 1)_{b_1}, [-C, 3)_{b_3}\}$. Here the subscripts denote the multiplicities of the intervals. From the constraints $a_0\leq b_1$ and $a_2\geq b_3$, we can construct a bijection $\sigma\colon B(M)\to B(N)$ by letting $\sigma([-C, \infty)) = [-C,\infty)$, $\sigma([-C,0)) = [-C,1)$, and $\sigma^{-1}([-C, 3)) = [-C,2)$. The remaining intervals in $B(M)$ are of the form $[-C,2)$ and are paired with the remaining intervals $[-C, 1)\in B(N)$. We conclude that 
\[ d_W^p(B(M), B(N)) \leq \text{$p$-cost}(\sigma) =  \left(\sum_{I\in B(M)} ||I-\sigma(I)||_p^p\right)^{1/p} = (n-1)^{1/p}.\]To see that this is optimal, we observe that any finite interval $I$ in $B(M)$ must be matched with an interval $J$ in $B(N)$ for which $d(I,J)\geq 1$; leaving a finite interval unmatched would result in a cost of at least $((C/2)^p + (C/2)^p)^{1/p} \geq 2^{(p+1)/p}n^{1/p}  \geq 2^{1/p}n^{1/p}$. Since there are $n-1$ finite intervals in $B(M)$, it follows that $\sigma$ is optimal, and 
\[ d_W^p(B(M), B(N)) = (n-1)^{1/p}.\qedhere\]
\begin{center}
\begin{figure}
\begin{tikzpicture}
\node [draw, shape = circle, fill = black, minimum size = 0.1cm, inner sep=0pt] (Mt) at (0,1){};
\node [draw, shape = circle, fill = black, minimum size = 0.1cm, inner sep=0pt] (Ml) at (-2,-1){};
\node [draw, shape = circle, fill = black, minimum size = 0.1cm, inner sep=0pt] (Mc) at (0,-1){};
\node [draw, shape = circle, fill = black, minimum size = 0.1cm, inner sep=0pt] (Mr) at (2,-1){};
\node [draw, shape = circle, fill = black, minimum size = 0.1cm, inner sep=0pt] (Mr-r) at (4,-1){};
\node (Mdots) at (3,-1) {$\cdots$};
\node (Mdots2) at (3,-4) {$\cdots$};
\draw (Mt) -- (Ml);
\draw (Mt) -- (Mc);
\draw (Mt) -- (Mr);
\draw (Mt) -- (Mr-r);
\draw[] (Mt) -- (Mdots);

\node[left] at (-2.5,-0.5) {$r_{1,1}$};
\node[left] at (-2.5,-1) {$r_{1,2}$};
\node[left] at (-2.5,-1.5) {$\vdots$};
\node[left] at (-2.5,-2) {$r_{1,s_1-1}$};

\draw[dashed] (-2.5,-0.5) -- (Ml);
\draw[dashed] (-2.5,-1) -- (Ml);
\draw[dashed] (-2.5,-1.5) -- (Ml);
\draw[dashed] (-2.5,-2) -- (Ml);

\node[left] at (-0.5,-0.5) {$r_{2,1}$};
\node[left] at (-0.5,-1) {$r_{2,2}$};
\node[left] at (-0.5,-1.5) {$\vdots$};
\node[left] at (-0.5,-2) {$r_{2,s_2-1}$};

\draw[dashed] (-0.5,-0.5) -- (Mc);
\draw[dashed] (-0.5,-1) -- (Mc);
\draw[dashed] (-0.5,-1.5) -- (Mc);
\draw[dashed] (-0.5,-2) -- (Mc);

\node[left] at (1.5,-0.5) {$r_{3,1}$};
\node[left] at (1.5,-1) {$r_{3,2}$};
\node[left] at (1.5,-1.5) {$\vdots$};
\node[left] at (1.5,-2) {$r_{3,s_3-1}$};

\draw[dashed] (1.5,-0.6) -- (Mr);
\draw[dashed] (1.5,-1) -- (Mr);
\draw[dashed] (1.5,-1.5) -- (Mr);
\draw[dashed] (1.5,-2) -- (Mr);

\node[right] at (4.5,-0.5) {$r_{m,1}$};
\node[right] at (4.5,-1) {$r_{m,2}$};
\node[right] at (4.5,-1.5) {$\vdots$};
\node[right] at (4.5,-2) {$r_{m,s_m-1}$};

\draw[dashed] (4.5,-0.5) -- (Mr-r);
\draw[dashed] (4.5,-1) -- (Mr-r);
\draw[dashed] (4.5,-1.5) -- (Mr-r);
\draw[dashed] (4.5,-2) -- (Mr-r);

\node[left] at (-1,2) {$r_1$};
\node[left] at (-1,1.5) {$r_2$};
\node[left] at (-1,1) {$\vdots$};
\node[left] at (-1,0.5) {$r_{m-1}$};

\draw[dashed] (-1,2) -- (Mt);
\draw[dashed] (-1,1.5) -- (Mt);
\draw[dashed] (-1,1) -- (Mt);
\draw[dashed] (-1,0.5) -- (Mt);

\node [draw, shape = circle, fill = black, minimum size = 0.1cm, inner sep=0pt, label=below:$g_{1,1}$] (Mll) at (-2.5,-4){};
\node (mlc) at (-2,-4) {$\cdots$};
\node [draw, shape = circle, fill = black, minimum size = 0.1cm, inner sep=0pt, label=below:$g_{1,s_1}$] (Mlr) at (-1.5,-4){};

\node [draw, shape = circle, fill = black, minimum size = 0.1cm, inner sep=0pt, label=below:$g_{2,1}$] (Mcl) at (-0.5,-4){};
\node (mcc) at (-0,-4) {$\cdots$};
\node [draw, shape = circle, fill = black, minimum size = 0.1cm, inner sep=0pt, label=below:$g_{2,s_2}$] (Mcr) at (0.5,-4){};

\node [draw, shape = circle, fill = black, minimum size = 0.1cm, inner sep=0pt, label=below:$g_{3,1}$] (Mrl) at (1.5,-4){};
\node (mrc) at (2,-4) {$\cdots$};
\node [draw, shape = circle, fill = black, minimum size = 0.1cm, inner sep=0pt, label=below:$g_{3,s_3}$] (Mrr) at (2.5,-4){};

\node [draw, shape = circle, fill = black, minimum size = 0.1cm, inner sep=0pt, label=below:$g_{m,1}$] (Mr-rl) at (3.5,-4){};
\node (mrrc) at (4,-4) {$\cdots$};
\node [draw, shape = circle, fill = black, minimum size = 0.1cm, inner sep=0pt, label=below:$g_{m,s_m}$] (Mr-rr) at (4.5,-4){};

\draw (Ml) -- (Mll);
\draw (Ml) -- (Mlr);

\draw (Mc) -- (Mcl);
\draw (Mc) -- (Mcr);

\draw (Mr) -- (Mrl);
\draw (Mr) -- (Mrr);

\draw (Mr-r) -- (Mr-rr);
\draw (Mr-r) -- (Mr-rl);

\draw[] (Mdots) -- (Mdots2);

\draw[-] (-4,-4.5) -- (-4,2.5);

\draw[-] (-4.1,-4) -- (-3.9, -4);
\draw[-] (-4.1,-1) -- (-3.9, -1);
\draw[-] (-4.1,1) -- (-3.9, 1);
\draw[-] (-4.1,2) -- (-3.9, 2);
\draw[-] (-4.1,0) -- (-3.9, 0);
\node[left] at (-4,-4) {$-C$};
\node[left] at (-4,-1) {$0$};
\node[left] at (-4,0) {$1$};
\node[left] at (-4,1) {$2$};
\node[left] at (-4,2) {$3$};
\end{tikzpicture}
\caption{The merge tree $M$ associated to an instance $(S,k)$.}
\label{fig:mergeM}
\end{figure}
\end{center}

\begin{center}
\begin{figure}
\begin{tikzpicture}
\node [draw, shape = circle, fill = black, minimum size = 0.1cm, inner sep=0pt] (Mt) at (0,2){};
\node [draw, shape = circle, fill = black, minimum size = 0.1cm, inner sep=0pt] (Ml) at (-2,0){};
\node [draw, shape = circle, fill = black, minimum size = 0.1cm, inner sep=0pt] (Mc) at (0,0){};
\node [draw, shape = circle, fill = black, minimum size = 0.1cm, inner sep=0pt] (Mr) at (2,0){};
\node [draw, shape = circle, fill = black, minimum size = 0.1cm, inner sep=0pt] (Mr-r) at (4,0){};
\node (Mdots) at (3,0) {$\cdots$};
\node (Mdots2) at (3,-4) {$\cdots$};
\draw (Mt) -- (Ml);
\draw (Mt) -- (Mc);
\draw (Mt) -- (Mr);
\draw (Mt) -- (Mr-r);
\draw[] (Mt) -- (Mdots);

\node[left] at (-2.5,-0.5) {$r'_{1,1}$};
\node[left] at (-2.5,-1) {$r'_{1,2}$};
\node[left] at (-2.5,-1.5) {$\vdots$};
\node[left] at (-2.5,-2) {$r'_{1,\AS-1}$};

\draw[dashed] (-2.5,-0.5) -- (Ml);
\draw[dashed] (-2.5,-1) -- (Ml);
\draw[dashed] (-2.5,-1.5) -- (Ml);
\draw[dashed] (-2.5,-2) -- (Ml);

\node[left] at (-0.5,-0.5) {$r'_{2,1}$};
\node[left] at (-0.5,-1) {$r'_{2,2}$};
\node[left] at (-0.5,-1.5) {$\vdots$};
\node[left] at (-0.5,-2) {$r'_{2,\AS-1}$};

\draw[dashed] (-0.5,-0.5) -- (Mc);
\draw[dashed] (-0.5,-1) -- (Mc);
\draw[dashed] (-0.5,-1.5) -- (Mc);
\draw[dashed] (-0.5,-2) -- (Mc);

\node[left] at (1.5,-0.5) {$r'_{3,1}$};
\node[left] at (1.5,-1) {$r'_{3,2}$};
\node[left] at (1.5,-1.5) {$\vdots$};
\node[left] at (1.5,-2) {$r'_{3,\AS-1}$};

\draw[dashed] (1.5,-0.6) -- (Mr);
\draw[dashed] (1.5,-1) -- (Mr);
\draw[dashed] (1.5,-1.5) -- (Mr);
\draw[dashed] (1.5,-2) -- (Mr);

\node[right] at (4.5,-0.5) {$r'_{k,1}$};
\node[right] at (4.5,-1) {$r'_{k,2}$};
\node[right] at (4.5,-1.5) {$\vdots$};
\node[right] at (4.5,-2) {$r'_{k,\AS-1}$};

\draw[dashed] (4.5,-0.5) -- (Mr-r);
\draw[dashed] (4.5,-1) -- (Mr-r);
\draw[dashed] (4.5,-1.5) -- (Mr-r);
\draw[dashed] (4.5,-2) -- (Mr-r);

\node[left] at (-1,3) {$r'_1$};
\node[left] at (-1,2.5) {$r'_2$};
\node[left] at (-1,2) {$\vdots$};
\node[left] at (-1,1.5) {$r'_{k-1}$};

\draw[dashed] (-1,3) -- (Mt);
\draw[dashed] (-1,2.5) -- (Mt);
\draw[dashed] (-1,1.5) -- (Mt);
\draw[dashed] (-1,1) -- (Mt);

\node [draw, shape = circle, fill = black, minimum size = 0.1cm, inner sep=0pt, label=below:$h_{1,1}$] (Mll) at (-2.5,-4){};
\node (mlc) at (-2,-4) {$\cdots$};
\node [draw, shape = circle, fill = black, minimum size = 0.1cm, inner sep=0pt, label=below:$h_{1,\AS}$] (Mlr) at (-1.5,-4){};

\node [draw, shape = circle, fill = black, minimum size = 0.1cm, inner sep=0pt, label=below:$h_{2,1}$] (Mcl) at (-0.5,-4){};
\node (mcc) at (-0,-4) {$\cdots$};
\node [draw, shape = circle, fill = black, minimum size = 0.1cm, inner sep=0pt, label=below:$h_{2,\AS}$] (Mcr) at (0.5,-4){};

\node [draw, shape = circle, fill = black, minimum size = 0.1cm, inner sep=0pt, label=below:$h_{3,1}$] (Mrl) at (1.5,-4){};
\node (mrc) at (2,-4) {$\cdots$};
\node [draw, shape = circle, fill = black, minimum size = 0.1cm, inner sep=0pt, label=below:$h_{3,\AS}$] (Mrr) at (2.5,-4){};

\node [draw, shape = circle, fill = black, minimum size = 0.1cm, inner sep=0pt, label=below:$h_{k,1}$] (Mr-rl) at (3.5,-4){};
\node (mrrc) at (4.25,-4) {$\cdots$};
\node [draw, shape = circle, fill = black, minimum size = 0.1cm, inner sep=0pt, label=below:$h_{k,\AS}$] (Mr-rr) at (5,-4){};

\draw (Ml) -- (Mll);
\draw (Ml) -- (Mlr);

\draw (Mc) -- (Mcl);
\draw (Mc) -- (Mcr);

\draw (Mr) -- (Mrl);
\draw (Mr) -- (Mrr);

\draw (Mr-r) -- (Mr-rr);
\draw (Mr-r) -- (Mr-rl);

\draw[] (Mdots) -- (Mdots2);

\draw[-] (-4,-4.5) -- (-4,2.5);

\draw[-] (-4.1,-4) -- (-3.9, -4);
\draw[-] (-4.1,-1) -- (-3.9, -1);
\draw[-] (-4.1,1) -- (-3.9, 1);
\draw[-] (-4.1,2) -- (-3.9, 2);
\draw[-] (-4.1,0) -- (-3.9, 0);
\node[left] at (-4,-4) {$-C$};
\node[left] at (-4,-1) {$0$};
\node[left] at (-4,0) {$1$};
\node[left] at (-4,1) {$2$};
\node[left] at (-4,2) {$3$};
\end{tikzpicture}
\caption{The merge tree $N$ associated to an instance $(S,k)$.}
\label{fig:mergeN}
\end{figure}
\end{center}

\end{proof}

\subsection{Proof of NP-hardness}
The goal of this section is to show the following. 
\begin{theorem}
Let $M$ and $N$ be merge trees associated to an instance $\Ss$ of the BAL-PART problem, and let $p\in [1,\infty)$.
Then $\Ss$ has a solution if and only if
\[d_I^p(M,N) \leq (n-1)^{1/p}.\]
\end{theorem}
\cref{thm:main-merge} is then an immediate consequence of this 
theorem and the NP-hardness of BAL-PART. Our proof proceeds in two steps.
First, we assume that $\Ss$ has a solution, and use this to conclude that $d_I^p(M,N) = (n-1)^{\frac 1p}$.
Then, we show that if $\Ss$ does not have a solution, then $d_I^p(M,N) \geq n^{\frac 1p}$.

\subsubsection{Step 1: If $\Ss$ has a solution, then $d_I^p(M,N) = (n-1)^{\frac 1p}$.}
We will show that if the BAL-PART instance $\Ss$ has a solution, then
\[d_I^p(M,N) = d_W^p(B(M), B(N))= (n-1)^{\frac 1p}.\]
By \cite[Theorem 1.2]{cardona2022universal}, $d_I^p(M,N) \geq d_W^p(B(M), B(N))$, so it suffices to show that
\[d_I^p(M,N) \leq (n-1)^{\frac 1p} \stackrel{\cref{lem:dwp}}{=} d_W^p(B(M), B(N)).\]

By assumption, there exists a partitioning of $S$ solving $\Ss = (S,k)$. By rearranging indices, we can assume that the sets of the partition are 
\[\{s_1,\ldots ,s_{i_1}\}, \{s_{i_1+1}, \ldots, s_{i_2}\}, \ldots, \{s_{i_{k-1}+1}, \ldots, s_m\}.\]
As illustrated in \cref{fig:mergeM} and \cref{fig:N3PART}, the presentations $P=(G_P, R_P, \gr_P)$ and $Q = (G_Q, R_Q, \gr_Q)$ in the definitions of $M$ and $N$ above are compatible.
This can easily be seen by renaming the generators as
\begin{align*}
(g_1,\dots,g_n) &\coloneqq (g_{1,1},\dots,g_{1,s_1},g_{2,1},\dots,g_{2,s_2},\dots,g_{m,1},\dots,g_{m,s_m}),\\
(h_1,\dots,h_n) &\coloneqq (h_{1,1},\dots,h_{1,\AS},h_{2,1},\dots,h_{2,\AS},\dots,h_{k,1},\dots,h_{k,\AS}).
\end{align*}
Then the relations in $R_P$ are $g_1=g_2, g_2=g_3, \dots, g_{n-1}=g_n$, and the relations of $R_Q$ are $h_1=h_2, h_2=h_3, \dots, h_{n-1}=h_n$.
It remains to show that
\[
|\gr_P(g_i=g_{i+1})-\gr_Q(h_i=h_{i+1})|\leq 1
\]
for all $i$. Since the grades of the relations in $R_P$ are either $0$ or $2$ and the grades of the relations in $R_Q$ are either $1$ or $3$, it is enough to show that if $\gr_P(g_i=g_{i+1}) = 0$, then $\gr_Q(h_i=h_{i+1}) = 1$. But if $\gr_P(g_i=g_{i+1}) = 0$, then $g_i=g_{i+1}$ is of the form $g_{k,\ell} = g_{k,\ell+1}$, and it follows from the partitioning provided by the solution of $\Ss$ that $h_i=h_{i+1}$ is of the form $h_{k',\ell'} = h_{k,\ell'+1}$, so $\gr_Q(h_i=h_{i+1}) = 1$.
\begin{center}
\begin{figure}
\begin{tikzpicture}
\node [draw, shape = circle, fill = black, minimum size = 0.1cm, inner sep=0pt] (Mt) at (0,2){};
\node [draw, shape = circle, fill = black, minimum size = 0.1cm, inner sep=0pt] (Mc) at (0,0){};
\node [draw, shape = circle, fill = black, minimum size = 0.1cm, inner sep=0pt] (Mr-r) at (4,0){};
\node (Mdots) at (2,0) {$\cdots$};
\node (Mdots2) at (3,-4) {$\cdots$};
\draw (Mt) -- (Mc);
\draw (Mt) -- (Mr-r);
\draw[] (Mt) -- (Mdots);

\node[left] at (-1.5,-0.5) {$r'_{1,1}$};
\node[left] at (-1.5,-1) {$r'_{1,2}$};
\node[left] at (-1.5,-1.5) {$\vdots$};
\node[left] at (-1.5,-2) {$r'_{1,t-1}$};

\draw[dashed] (-1.5,-0.5) -- (Mc);
\draw[dashed] (-1.5,-1) -- (Mc);
\draw[dashed] (-1.5,-1.5) -- (Mc);
\draw[dashed] (-1.5,-2) -- (Mc);

\node[right] at (4.5,-0.5) {$r'_{k,1}$};
\node[right] at (4.5,-1) {$r'_{k,2}$};
\node[right] at (4.5,-1.5) {$\vdots$};
\node[right] at (4.5,-2) {$r'_{k,\AS-1}$};

\draw[dashed] (4.5,-0.5) -- (Mr-r);
\draw[dashed] (4.5,-1) -- (Mr-r);
\draw[dashed] (4.5,-1.5) -- (Mr-r);
\draw[dashed] (4.5,-2) -- (Mr-r);

\node[left] at (-1,2) {$r'_1$};
\node[left] at (-1,1.5) {$r'_2$};
\node[left] at (-1,1) {$\vdots$};
\node[left] at (-1,0.5) {$r'_{k-1}$};

\draw[dashed] (-1,2) -- (Mt);
\draw[dashed] (-1,1.5) -- (Mt);
\draw[dashed] (-1,1) -- (Mt);
\draw[dashed] (-1,0.5) -- (Mt);

\node [draw, shape = circle, fill = black, minimum size = 0.1cm, inner sep=0pt, label=below:$h_{1,1}$] (Mll) at (-2.5,-4){};
\node (mlc) at (-2,-4) {$\cdots$};
\node [draw, shape = circle, fill = black, minimum size = 0.1cm, inner sep=0pt, label=below:$h_{1,s_1}$] (Mlr) at (-1.5,-4){};

\node [draw, shape = circle, fill = black, minimum size = 0.1cm, inner sep=0pt, label=below:$h_{1,s_1+1}$] (Mcl) at (-0.5,-4){};
\node (mcc) at (0.25,-4) {$\cdots\cdots$};
\node [draw, shape = circle, fill = black, minimum size = 0.1cm, inner sep=0pt, label=below:$h_{1,s_1+s_2}$] (Mcr) at (1,-4){};

\node [draw, shape = circle, fill = black, minimum size = 0.1cm, inner sep=0pt, label=below:$ $] (Mrl) at (1.5,-4){};
\node (mrc) at (2,-4) {$\cdots$};
\node [draw, shape = circle, fill = black, minimum size = 0.1cm, inner sep=0pt, label=below:$h_{1,\AS}$] (Mrr) at (2.5,-4){};

\node [draw, shape = circle, fill = black, minimum size = 0.1cm, inner sep=0pt, label=below:$h_{k,1}$] (Mr-rl) at (3.5,-4){};
\node (mrrc) at (4.25,-4) {$\cdots$};
\node [draw, shape = circle, fill = black, minimum size = 0.1cm, inner sep=0pt, label=below:$h_{k,\AS}$] (Mr-rr) at (5,-4){};

\draw (Mc) -- (Mll);
\draw (Mc) -- (Mlr);

\draw (Mc) -- (Mcl);
\draw (Mc) -- (Mcr);

\draw (Mc) -- (Mrl);
\draw (Mc) -- (Mrr);

\draw (Mr-r) -- (Mr-rr);
\draw (Mr-r) -- (Mr-rl);

\draw[] (Mdots) -- (Mdots2);

\draw[-] (-4,-4.5) -- (-4,2.5);

\draw[-] (-4.1,-4) -- (-3.9, -4);
\draw[-] (-4.1,-1) -- (-3.9, -1);
\draw[-] (-4.1,1) -- (-3.9, 1);
\draw[-] (-4.1,2) -- (-3.9, 2);
\draw[-] (-4.1,0) -- (-3.9, 0);
\node[left] at (-4,-4) {$-C$};
\node[left] at (-4,-1) {$0$};
\node[left] at (-4,0) {$1$};
\node[left] at (-4,1) {$2$};
\node[left] at (-4,2) {$3$};
\end{tikzpicture}
\caption{The merge tree $N$, assuming a solution to BAL-PART.}
\label{fig:N3PART}
\end{figure}
\end{center}
\subsubsection{Step 2: If $\Ss$ has no solution, then $d_I^p(M,N) \geq n^{\frac 1p}$.}
\label{sec_merge_l_bound}

We will now assume that $\Ss$ does not have a solution and show that $d_I^p(M,N) \geq n^{\frac 1p}$.

For a presentation $P=(G,R,\gr)$ of a merge tree $M$, we define $\mrg_P\colon G\times G\to \R$ by \[\mrg_P(g_1,g_2) = \min \{s \mid \gr(g_1)\geq s, \gr(g_2)\geq s \text{, and } g_1\sim_s g_2\},\]
where $\sim_s$ is defined in \cref{sec:basic}. Geometrically, $\mrg_{P}(g_1,g_2)$ is the earliest point at which the branches containing $g_1$ and $g_2$ merge. This minimum always exists.

\begin{lemma}
\label{lem:bijection-3part}
Let $P=(G,R,\gr)$ and $Q=(H, R', \gr')$ denote the presentations given for $M$ and $N$ associated to the instance $(S,k)$ of BAL-PART. If there exists a bijection $\sigma\colon G\to H$ such that $\mrg_Q(\sigma(g_1), \sigma(g_2))=1$ for all $g_1, g_2\in G$ for which $\mrg_P(g_1, g_2) = 0$, then $(S,k)$ has a solution.
\end{lemma}
\begin{proof}
For any $1\leq j\leq m$ and $1<u\leq s_j$, $\mrg_P(g_{j,1}, g_{j,u})=0$, and therefore \[\mrg_Q(\sigma(g_{j,1}), \sigma(g_{j,u})) = 1\] by assumption.  In particular, there exists an $1\leq l \leq k$ such that 
\[\{\sigma(g_{j,1}), \sigma(g_{j,2}), \ldots, \sigma(g_{j, s_j})\} \subseteq \{h_{l,1}, \ldots, h_{l, \hat{S}}\}.\]
Since $\sigma$ is a bijection, this immediately gives a solution of $(S,k)$.
\end{proof}

\paragraph{High-level approach}Before continuing, we give an informal overview of the remainder of the section.
We show in \cref{lem_subdivide_merge} that when connecting $M$ to $N$ through merge trees $M_i$, we can choose matchings between $B(M_i)$ and $B(M_{i+1})$ such that the change in the right endpoint of an interval corresponds to a change in the degree of a relation in the two compatible presentation matrices.
Moreover, the composition of matchings \[B(M)\to B(M_2) \to B(M_3)\to \cdots\] is injective, and induces a bijection between $B(M)$ and $B(N)$ (\cref{lem:matching}).
In particular, an endpoint in $B(M)$ at $0$ or $2$ must ``travel'' to $1$ or $3$, and therefore travel a distance of at least 1. Similarly, \cref{lem:bijection-3part} and \cref{lem:sigma-comp} tell us that at least one relation $r$ must ``travel'' from 0 to 3.
The idea is then to show that the cost associated to the two compatible presentations of $M_i$ and $M_{i+1}$ is at least $||(d_1, \ldots, d_{n-1}, \delta)||_p$ where $d_j$ is the change in the right endpoint of the $j$-th bar in $B(M)$ as it moves from $B(M_i)$ to $B(M_{i+1})$ through the matching, and $\delta$ denotes the change in grade of the tracked relation $r$ when going from $M_i$ to $M_{i+1}$. Adding this up, with the help of the triangle inequality, as we traverse from $M$ to $N$, gives  $d_I^p(M,N)+\epsilon > ||(1,\ldots, 1)||_p = n^{\frac 1p}$.
Here $\epsilon>0$ is a sufficiently small constant; recall that the $p$-presentation matrix is defined through an infimum.  However, this approach does not work on the nose because the effect of moving the relation $r$ will already be taken into account in the change of the right endpoint of an interval. Care must therefore be taken to ensure that we avoid double counting. To address this, we count only the contribution from right endpoints moving in a certain way, and we apply a complementary approach for the relation. The thick blue and red lines in \cref{fig:track} illustrate this. The change in the right endpoint from \(M_3\) to \(M_4\) is added to the total cost, whereas the same type of movement (from 2 in the direction of 1) of the relation from $M_4$ to $M_5$ is ignored.  The notation in the caption will become clear later.

To proceed, we first need the following technical lemma. 

\begin{figure}
\centering
\begin{tikzpicture}[scale=1.5]
\node at (0,1){$P_1$};
\draw[shorten <=.3cm, shorten >=.3cm] (0,1) to (.6,0);
\node at (.32,0){$M = M_1$};
\draw[shorten <=.3cm, shorten >=.3cm] (1.2,1) to (0.6,0);
\node at (1.2,1){$P'_1$};
\draw (1.4,1) to (2.4,1);
\node at (1.9,1.4){same gens.};
\node at (1.9,1.14){and rels.};
\node at (2.6,1){$P_2$};
\draw[shorten <=.3cm, shorten >=.3cm] (2.6,1) to (3.2,0);
\node at (3.2,0){$M_2$};
\draw[shorten <=.3cm, shorten >=.3cm] (3.8,1) to (3.2,0);
\node at (3.8,1){$P'_2$};
\node at (4.9,.5){$\dots$};
\begin{scope}[xshift=5cm]
\node at (1,0){$M'_{\ell-1}$};
\draw[shorten <=.3cm, shorten >=.3cm] (1.3,1) to (1,0);
\node at (1.3,1){$P'_{\ell-1}$};
\draw (1.6,1) to (2.5,1);
\node at (2.05,1.4){same gens.};
\node at (2.05,1.14){and rels.};
\node at (2.7,1){$P_\ell$};
\draw[shorten <=.3cm, shorten >=.3cm] (2.7,1) to (3,0);
\node at (3.3,0){$M_\ell = N$};
\end{scope}
\end{tikzpicture}
\caption{The ``zigzag'' of merge trees and presentations in \cref{lem_refinement}, going from $M$ on the left to $N$ on the right, in each step changing from one presentation of $M_i$ to another and then changing only the grade function of the presentation.}
\label{fig_M_i_P_i_zz}
\end{figure}

\begin{lemma}
\label{lem_refinement}
For every $\epsilon>0$, there are merge trees $M=M_1, M_2, \dots, M_\ell=N$ and presentations $P'_1,P_2,P'_2,\dots, P'_{\ell-1},P_\ell$ such that
\begin{equation}
\label{eq_pres_seq}
d_I^p(M,N) + \epsilon > \sum_{i=1}^{\ell-1} d^p(P_i', P_{i+1})
\end{equation}
and such that the following conditions are satisfied,
\begin{itemize}
\item[(i)] $P_i$ and $P'_i$ are presentations of $M_i$,
\item[(ii)] $P'_i = (G_i,R_i,\gr'_i)$ and $P_{i+1} = (G_i,R_i,\gr_{i+1})$ for some $G_i$ and $R_i$ (so $P'_i$ and $P_{i+1}$ have the same sets $G_i$ and $R_i$ of generators and relations),
\item[(iii)] we can equip $G_i\cup R_i$ with a total order $\leq_i$ such that for any $c\leq_i c'\in G_i\cup R_i$, we have $\gr'_i(c)\leq \gr'_i(c')$ and $\gr_{i+1}(c)\leq \gr_{i+1}(c')$,
\item[(iv)] we do not have $\gr'_i(c) < z < \gr_{i+1}(c)$ or $\gr'_i(c) > z > \gr_{i+1}(c)$ for any $c\in G_i\cup R_i$, $z\in \{0,1,2,3\}$.
\end{itemize}
\end{lemma}
\begin{proof}
By the definition of $d_I^p$, we have merge trees $M=M_1, M_2, \dots, M_\ell=N$ and presentations $P'_1,P_2,P'_2,\dots, P'_{\ell-1},P_\ell$ such that  \cref{eq_pres_seq} holds, and such that (i) and (ii) are satisfied; see also \cref{fig_M_i_P_i_zz}. To obtain (iii) and (iv), we shall refine the sequence of merge trees through interpolation. To do so, we first fix an integer $1\leq i\leq \ell-1$, and let $P' = P'_i$ and $P = P_{i+1}$. We denote the grades of $P'$ and $P$ by $\gr'$ and $\gr$, respectively.
For every $\lambda \in [0,1]$, we get a presentation matrix $P^\lambda = (G,R,\gr_\lambda)$ where $\gr_\lambda(c) = (1-\lambda)\gr'(c)+\lambda \gr(c)$.
Note that $P^0 = P'$ and $P^1 = P$. We slightly extend the domain of the grading functions by letting $\gr_\lambda(j) = j$ for $j\in \{0,1,2,3\}$ and for all $\lambda$. Then, for any two $c, c'\in G\cup R\cup \{0,1,2,3\}$, we have that either
\begin{enumerate}
    \item $\gr_\lambda(c)\leq \gr_{\lambda}(c')$ for all $\lambda$, or
    \item $\gr_{\lambda}(c')\leq \gr_{\lambda}(c)$ for all $\lambda$, or
    \item there exists a unique $\lambda$ such that $\gr_\lambda(c) = \gr_{\lambda}(c')$ and the sign of $\gr_\lambda(c)-\gr_\lambda(c')$ is locally constant on $[0,\lambda)\cup (\lambda,1]$. If this is the case, then we say that $c$ and $c'$ \emph{cross} and let $\lambda(c,c') = \lambda$ denote the unique crossing time. 
\end{enumerate}
If $c$ and $c'$ do not cross, then we let $\lambda(c,c') = 0$. 
Now, let $T$ be the set 
\[ T = \{0,1\}\cup \bigcup_{c,c'\in G\cup R\cup\{{0,1,2,3\}}} \lambda(c,c')\]
By ordering the elements $\{\lambda_1, \ldots, \lambda_{e_i}\}$ of $T$ according to the standard total order on the real numbers, we get a finite sequence of real numbers
\[0=\lambda_1<\dots<\lambda_{e_i}=1\]such that the sequence of merge trees $M_i=M^i_1, M^i_2, \ldots, M^i_{e_i} = M_{i+1}$ given by the presentation matrices
\[P'=(P_1^i)', P_2^i, (P_2^i)', \ldots, (P_{e_i-1}^i)', P_{e_i}^i=P\]
with $P_j^i = (P_j^i)'=(G,R,\gr_{\lambda_j})$, $1\leq j\leq e_i$, satisfy (iv). To ensure that the merge trees and presentations also satisfy (iii), let $\leq_{j}$ be any total order on $G\cup R$  satisfying 
\begin{enumerate}
    \item $c<_j c'$ if  $\gr_{\lambda_j}(c) < \gr_{\lambda_j}(c')$, or
    \item $c<_j c'$ if $\gr_{\lambda_j}(c) = \gr_{\lambda_j}(c')$ and $\gr_{\lambda_{j+1}}(c) < \gr_{\lambda_{j+1}}(c')$.
\end{enumerate}
Order the elements of $G\cup R$ arbitrarily as $\{h_1, h_2, \ldots, h_{|R|+|G|}\}$ and for $1\leq j\leq e_i$ let $L(j)\in \mathbb{R}^{|G|+|R|}$ be the vector 
\[ L(j) = \begin{bmatrix} \gr_{\lambda_j}(h_1) \\ \gr_{\lambda_j}(h_2) \\ \hdots \\ \gr_{\lambda_j}(h_{|R|+|G|})\end{bmatrix},\]
and observe that $L(j) = (1-\lambda_j)L(1) + \lambda_jL(e_i)$. 
Since the presentation matrices are compatible for all $\lambda_j$, and we are interpolating the indices in the presentation matrices, we have that
\begin{linenomath}\begin{align}d^p(P',P) &= \left(\sum_{h\in G\cup R} |\gr_P(h)-\gr_{P'}(h)|^p\right)^{\frac 1 p} \label{eq:interpolation}\\
& = ||L(e_i) - L(1)||_p  = \sum_{j=1}^{e_i-1} ||L(e_{j+1}) - L(e_j)||_p \nonumber\\
&=\sum_{j=1}^{e_i-1} d^p((G,R,\gr_{\lambda_j}),(G,R,\gr_{\lambda_{j+1}}))  = \sum_{j=1}^{e_i-1} d^p((P_j^i)',P_{j+1}^i).\nonumber
\end{align}\end{linenomath}
Doing this for every $i$ gives the sequence of merge trees
\[M_1 = M_1^1, M_2^1, \ldots, M_{e_1}^1 = M_2 =  M_{1}^2,  \ldots, M_{e_2}^2 = M_3 = M_1^3, \ldots, M_{e_{\ell-1}}^{\ell-1} = M_\ell\]
and presentation matrices 
\[P_1 = (P_1^1)', P_2^1, (P_2^1)', \ldots, (P_{e_1-1}^1)', P_{e_1}^1 = P_2, P_2' = (P_1^2)',\ldots, P_{e_{\ell-1}}^{\ell-1}=P_l\]
satisfying (i) -- (iv) of the lemma. \cref{eq:interpolation} ensures that \cref{eq_pres_seq} is still satisfied.
\end{proof}

In the following, $B_0^{\rm fin}(M_i)$ denotes the subset of $B(M_i)$ of intervals of finite length. 
\begin{lemma}
\label{lem_subdivide_merge}
Suppose $M_1, M_2, \dots, M_\ell$ and $P'_1,P_2,P'_2,\dots, P'_{\ell-1},P_\ell$ satisfy the conditions of \cref{lem_refinement}. Then, for every $1\leq i<\ell$, there is a subset $\Pi_i \subseteq G_i\times R_i$, with each $g\in G_i$ and $r\in R_i$ showing up in at most one pair of $\Pi_i$, such that
\begin{linenomath}\begin{align*}
B_0^{\rm fin}(M_i) &= \{[\gr'_i(g),\gr'_i(r))\mid (g,r)\in \Pi_i\},\\
B_0^{\rm fin}(M_{i+1}) &= \{[\gr_{i+1}(g),\gr_{i+1}(r))\mid (g,r)\in \Pi_i\}.
\end{align*}
\end{linenomath}
\end{lemma}
\begin{proof}
Let $(G,R,\gr)$ be a presentation of a merge tree $\widehat{M}$, and let $\leq$ be a total order on $G\cup R$ such that $\gr(c)\leq \gr(c')$ whenever $c\leq c'$ for $c,c'\in G\cup R$. Seeing every generator $g\in G$ as a point, and a relation  $\{g,h\}\in R$ as an edge connecting $g$ to $h$, running the ``standard algorithm'' \cite[Ch.~VII.1]{edelsbrunner2022computational} for persistent homology on this sequence of points and vertices computes the persistence barcode of a filtered multigraph (two points may be connected by multiple edges). If no relation appears twice, however, then the computation will give precisely the persistence barcodes of $B(\widehat{M})$. Furthermore, if a relation appears twice, then the corresponding column will be cleared out and thus only contribute to the barcode in degree 1. In particular, the bars in $B_0(\widehat{M})$ correspond precisely to a set $\Pi\subseteq G\times R$, and this pairing depends only on the ordering of the generators and relations. We expand on these ideas and give an example in \cref{rem.stdalg}.

To prove the statement, note that, by virtue of Lemma 4 (iii), that the order of the generators and relations coincide for $P_i'$ and $P_{i+1}$. In particular, running the persistence algorithm on either $(G_i, R_i, \gr'_i)$ or $(G_i, R_i, \gr_{i+1})$, we obtain a set $\Pi_i \subseteq G_i\times R_i$ such that
\begin{linenomath}\begin{align*}
B^{\rm fin}_0(M_i) &= \{[\gr'_{\lambda_i}(g),\gr'_{\lambda_i}(r))\mid (g,r)\in \Pi_i\},\\
B^{\rm fin}_0(M_{i+1}) &= \{[\gr_{\lambda_{i+1}}(g),\gr_{\lambda_{i+1}}(r))\mid (g,r)\in \Pi_i\}.
\end{align*}\end{linenomath}
Note that each $g$ and $r$ is either paired or unpaired by the standard algorithm. In either case, any generator or relation appears in at most one pair. 
\end{proof}
The following result is immediate. 
\begin{corollary}
\label{cor_matching} Suppose $M_1, M_2, \dots, M_\ell$ and $P'_1,P_2,P'_2,\dots, P'_{\ell-1},P_\ell$ satisfy the conditions of \cref{lem_refinement}. Then, $\Pi_i$ defines a matching $\phi_i\colon B(M_i)\to B(M_{i+1})$ where the two infinite bars are matched to each other, and such that 
    \[\phi_i([\gr_i'(g), \gr_i'(r))) = [\gr_{i+1}(g), \gr_{i+1}(r)), \qquad (g,r)\in \Pi_i\]
    if both of the two intervals are non-trivial, and all other intervals are unmatched. 
\end{corollary}

\begin{figure}
\centering
\begin{tikzpicture}

\begin{scope}[xshift=0cm,rotate around={90:(0,0)}]
 \draw[gray!40, dashed] (1,1) grid (8,-12);

         \foreach \x in {1, 2,3,4,5,6,7}{
    \fill[black] (\x,-4) circle (2pt) node[left] {};
    }
             \foreach \x in {2,3,4,5,6,7}{
    \fill[black] (\x,-5) circle (2pt) node[left] {};
    }
      \fill[black] (8,-4) circle (2pt) node[above] {$g_1$};
     \fill[black] (8,-5) circle (2pt) node[above] {$g_2$};
      \fill[black] (8,-6) circle (2pt) node[above] {$g_3$};
     \foreach \x in {3,4,5,6,7}{
    \fill[black] (\x,-6) circle (2pt) node[left] {};
    }
     \foreach \x in {4,5,6,7,8}{
    \draw[black] (\x,-4) -- (\x,-5);
    }
         \foreach \x in {5,6,7,8}{
    \draw[bend right=70]    (\x,-4) to (\x,-6);
    }
             \foreach \x in {7,8}{
    \draw[bend right=50]    (\x,-4) to (\x,-5);
    }

        \draw[thick,->] (1,-11) -- (8,-11) node[above] {$|M|$};
        \draw[thick] (2,-11.5) -- (4,-11);
        \draw[thick] (3,-12) -- (5,-11);

        \foreach \x/\y in {4/-8.5, 5/-9} {
        \draw (\x,\y) circle (3pt);
    }
            \foreach \x/\y in {1/-8, 2/-8.5, 3/-9} {
        \fill (\x,\y) circle (3pt);
    }
    \draw[->, thick] (1,-8) -- (8,-8) node[above] {$B(M)$};
    \draw[-, thick] (2,-8.5) -- (4,-8.5);
    \draw[-, thick] (3,-9) -- (5,-9);
        
        \draw[->, thick, blue] (1,1) -- (8,1) node[above] {$g_1$};
        \draw[->, thick, blue] (2,0) -- (8,0) node[above] {$g_2$};;
        \draw[->,thick, blue] (3,-1) -- (8,-1) node[above] {$g_3$};

        \fill[red] (4,-2) circle (3pt) node[right] {$r_1$};
        \fill[red] (5,-2) circle (3pt) node[right] {$r_2$};
        \fill[red] (7,-2) circle (3pt) node[right] {$r_3$};

        \draw[red, thick, bend right] (4,-2) to (4,1);
        \draw[red, thick, bend left] (4,-2) to (4,0);
         \draw[red, thick, bend right] (5,-2) to (5,1);
        \draw[red, thick, bend left] (5,-2) to (5,-1);

        \draw[red, thick, bend right, dashed] (7,-2) to (7,1);
        \draw[red, thick, bend left, dashed] (7,-2) to (7,0);

    \node[below] at (0.5, 0) {\textbf{A}};
     \node[below] at (0.5, -4) {\textbf{B}};
      \node[below] at (0.5, -8) {\textbf{C}};
       \node[below] at (0.5, -11) {\textbf{D}};
        
\end{scope}
\end{tikzpicture}
\caption{\textbf{(A)}: A presentation $(G,R)$ of a merge tree $M$. \textbf{(B)}: The corresponding filtered multigraph; see \cref{rem.stdalg}. \textbf{(C)}: The barcode $B(M)$ equals the barcode of the $0$-th persistent homology of the filtered multigraph. \textbf{(D)}: The geometric realization of $M$.}
\label{fig.stdalg}
\end{figure}

\begin{remark}
\label{rem.stdalg}
    For any presentation $(G,R)$ of a merge tree $M$ and any filtration value $t$, we can form a multigraph $\mathbf{G}_t$ with vertices all elements $g\in G$ with $\gr(g) \leq t$ and an edge between $g_i$ and $g_j$ for every $r=\{g_i, g_j\}$ with $\gr(r)\leq t$. Since two generators are equivalent in $M_t$ if and only if they are in the same connected component of $\mathbf{G}_t$, we obtain $B(M)$ from the $0$-th persistent homology of the filtration $\{\mathbf{G}_t\}_{t\in \R}.$ This is illustrated in \cref{fig.stdalg}.

    The persistent homology of $\{\mathbf{G}_t\}_{t\in \R}$ can be computed using the standard algorithm: order the vertices and edges by their appearance, form the corresponding boundary matrix with one simplex per row and column, and reduce the columns from left to right; see, e.g., \cite[Figure 7]{otter2017roadmap}.  If there is more than one edge between a pair of vertices, then this will introduce a loop and manifest itself as a 0 column in the reduced matrix (a feature in degree 1 homology). The persistence pairs, one for each pivot column, correspond to the finite bars $B(M)_0^{\rm fin}\subset B(M)$. In addition, $B(M)$ contains the interval $[\gr(g_1), \infty)$ where $\gr(g_1)$ is the minimal grade of all generators in the presentation.

    Discarding columns in the boundary matrix corresponding to vertices (those columns are all zero), the degree 0 persistent homology computation of $\mathbf{G}_t$ amounts to bringing a particular matrix representation of $(G,R)$ to column-reduced form (we are working over $\mathbb{Z}_2$ for simplicity, as $H_0$ computations are always independent of the field.) For the example in \cref{fig.stdalg}, we have

\begin{equation*}
\begin{bNiceMatrix}[
  first-row,code-for-first-row=\scriptstyle,
  first-col,code-for-first-col=\scriptstyle,
]
& \{g_1, g_2\}& \{g_1, g_3\} & \{g_1, g_2\} \\
g_1 & 1 & 1 & 1 \\
g_2 & 1 & 0 & 1 \\
g_3 & 0 & 1 & 0
\end{bNiceMatrix}\qquad \overset{\text{Adding 1st column to 3rd column}}{\longrightarrow} \qquad \begin{bNiceMatrix}[
  first-row,code-for-first-row=\scriptstyle,
  first-col,code-for-first-col=\scriptstyle,
]
& \{g_1, g_2\}& \{g_1, g_3\} & \{g_1, g_2\} \\
g_1 & 1 & 1 & 0 \\
g_2 & 1 & 0 & 0 \\
g_3 & 0 & 1 & 0
\end{bNiceMatrix}
\end{equation*}
Hence, we obtain the persistence pairs $(g_2, \{g_1, g_2\}), (g_3, \{g_1, g_3\})$ and thus
\[B(M)_0^{\rm fin}=\left\{\left[\gr(g_2), \gr(\{g_1, g_2\})\right), [\gr(g_3), \gr(\{g_1, g_3\}))\right\}.\] 
\end{remark}

\begin{lemma}
\label{lem:matching}
For $0 < \epsilon < d_I^p(M,N)$ and $\phi_i$ as in \cref{lem_refinement}, the composition of matchings 
$\phi = \phi_{\ell-1}\circ \cdots \circ  \phi_1 \colon B(M)\to B(N)$ is a bijection. 
\end{lemma}
\begin{proof}
Suppose that an interval $[-C,b)\in B(M)$ is left unmatched by $\phi$.
This contributes at least $\frac C2$ to $p\text{-cost}(\phi)$, so using \cref{cor_matching} and that the $p$-cost of a composition of morphisms is at most the sum of the $p$-costs of each morphism, we get
\[
\frac C2 \leq \text{$p$-cost}(\phi) \leq \sum_{i=1}^{\ell-1} \text{$p$-cost}(\phi_i) \leq  \sum_{i=1}^{\ell-1} d^p(P_i', P_{i+1}) < d_I^p(M,N)+\epsilon < 2n^{\frac 1p},
\]
which contradicts $C=4(n-1)^{\frac 1p}$.
\end{proof}

Let $P = (G,R,\gr)$ be a presentation of a merge tree.
For $g,h\in G$, recall that
\[
\mrg_{P}(g,h) = \min\{s\in \R\mid [g] = [h] \in G_s/{\sim}_s\}.
\]
and that $\mrg_{P}(g,h)$ is the earliest point at which the branches containing $g$ and $h$ merge. 
\begin{lemma}
\label{lem_long_branch_merge_levels}
Suppose $M_1, M_2, \dots, M_\ell$ and $P'_1,P_2,P'_2,\dots, P'_{\ell-1},P_\ell$ satisfy 
the conditions of \cref{lem_refinement}, and let $g,g'\in G_i$.
Then, there is an $r\in R_i$ such that $\mrg_{P_i'}(g,g') = \gr_i(r)$ and $\mrg_{P_{i+1}}(g,g') = \gr_i'(r)$.
\end{lemma}
\begin{proof}
For $r\in R$, let $\sim_r$ be the equivalence relation on $G$ generated by $\{h\sim h'\mid \{h,h'\}\leq_i r\}$, where $\leq_i$ satisfies (iii) in \cref{lem_refinement}. Since the order is a total order, there exists a unique $r\in R$ such that $g \sim_r g'$ but $g \not\sim_{r'} g'$ for any relation $r'<_i r$ in $R$. In particular, we have that $\mrg_{P_i'}(g,g') = \gr_i(r)$ and $\mrg_{P_{i+1}}(g,g') = \gr_i'(r)$.
\end{proof}

We shall now construct a bijection between the generators in the canonical presentations of $M$ and $N$ given in \cref{sec:3part}. To this end, assume that we have a sequence of merge trees $M_i$ and presentations $P_i, P_i'$ as described in \cref{lem_refinement}. Furthermore, let $P_1 = (G_0,R_0,\gr_0)$ and $P'_\ell = (G_{\ell},R_{\ell},\gr_{\ell})$ denote the canonical presentations of $M$ and $N$. Now, for all $1\leq i\leq \ell$ fix an isomorphism 
$\alpha_i \colon F(P_i)\to F(P_i'),$
and for every $1\leq i\leq \ell$, let $\sigma_{i-1}\colon G_{i-1}\to G_{i}$ be any map such that the following diagram commutes,
\[
\begin{tikzcd}
G_{i-1}\ar[d, "\text{proj}"] \ar[r, "\sigma_{i-1}"] & G_{i}\ar[d, "\text{proj}"]\\
F(P_{i}) \ar[r, "\alpha_{i}"] & F(P_{i}')
\end{tikzcd}
\]
This is always possible, as the vertical maps are simply projections into the equivalence classes. Furthermore, for any two $g,h\in G_{i-1}$, we have that 
\begin{equation}
\label{eq:equalmerge}
\mrg_{P_i}(g,h) = \mrg_{P_i'}(\sigma_{i-1}(g), \sigma_{i-1}(h)).
\end{equation}To see this, observe that if if $[g] = [h] \in (G_{i-1})_s$, then $[\sigma_{i-1}(g)] = \alpha_s[g] = \alpha_s[h] = [\sigma(h)]$. Therefore, $g\sim_s h \Rightarrow \sigma(g)\sim_s \sigma(h)$. The symmetric argument with $\alpha^{-1}$ shows that if $\sigma(g)\sim_s \sigma(h)$, then $g\sim_s h$.

\begin{lemma}
\label{lem:sigma-comp}
The composition $\sigma = \sigma_{\ell-1}\circ \cdots \circ \sigma_1 \colon G_1\to G_{\ell}$ is a bijection.
\end{lemma}
\begin{proof}
Since $G_0$ and $G_{\ell}$ contain the same number of elements, it suffices to show that $\sigma$ is injective. For any $g,h\in G_i$, seen as generators in the presentation $P_i'$ of $M_i$, we define \[\mu'(g,h) = \mrg_{P_i'}(g,h) - \max\{\gr'_i(g),\gr'_i(h)\}.\] 
Likewise, for $g,h\in G_{i}$, seen as generators in the presentation $P_{i+1}$ of $M_{i+1}$, we let  \[\mu(g,h) = \mrg_{P_{i+1}}(g,h) - \max\{\gr_{i+1}(g),\gr_{i+1}(h)\}.\]
Consider any two generators $g, h\in G_0$ such that $\mrg_M(g,h) =0$, and inductively define $g_i,h_i\in G_i$ by $g_{i+1} = \sigma_i(g_i)$ and $h_{i+1} = \sigma_i(h_i)$.
By construction, $\mu(g,h) = C$, see \cref{fig:mergeM}, and note that 
\[\mrg_{P_{i}'}(g_i,h_i) = \mrg_{P_{i+1}}(g_{i+1},h_{i+1}),\] 
by \cref{eq:equalmerge}, and also that  $\gr_{i+1}(g_i)\geq \gr'_{i+1}(g_{i+1})$ and $\gr_{i+1}(h_i)\geq \gr'_{i+1}(h_{i+1})$. Hence,  $\mu(g_i,h_i)\leq \mu'(g_{i+1},h_{i+1})$.

We know from \cref{lem_long_branch_merge_levels} that there is an $r$ such that $\mrg_{P'_i}(g_i,h_i) = \gr'_i(r)$ and $\mrg_{P_{i+1}}(g_i,h_i) = \gr_{i+1}(r)$.
We get
\begin{align*}
&\mu(g_{i-1}, h_{i-1}) - \mu(g_i,h_i)  \leq \mu'(g_i,h_i) - \mu(g_i,h_i) \\ 
&=\gr_i'(r) - \max\{\gr'_i(g_i),\gr'_i(h_i)\} - (\gr_{i+1}(r) - \max\{\gr_{i+1}(g),\gr_{i+1}(h)\})\\
&= \gr_i'(r) - \gr_{i+1}(r) + \max\{\gr_{i+1}(g),\gr_{i+1}(h)\}) - \max\{\gr'_i(g_i),\gr'_i(h_i)\}\\
&\leq |\gr'_i(r)-\gr_{i+1}(r)| + \max\{|\gr'_i(g)-\gr_{i+1}(g)|, |\gr'_i(h)-\gr_{i+1}(h)|\}\\  &\leq 2d^p(P_i', P_{i+1}).
\end{align*}
For the first inequality we have used that $\max\{a,b\}-\max\{c,d\}\leq \max\{|a-c|, |b-d\}$ for all real numbers $a,b,c,d$. Combined with $\mu(g_i,h_i)\leq \mu'(g_{i+1},h_{i+1})$ for all $i$, we get
\[
\mu(g_0,h_0) - \mu'(g_\ell,h_\ell) \leq 2\sum_{i=1}^{\ell-1} d^p(P'_i,P_{i+1}) \leq 2d_I^p(M,N) < C.
\]
Since $\mu(g_0,h_0) = C$, we have $\mu'(g_\ell,h_\ell)>0$, which means that $g_\ell \neq h_\ell$.
It follows that $\sigma$ is injective and thus a bijection. 
\end{proof}

In order to conclude the proof, we must introduce some additional notation; see \cref{fig:track} for an illustration. First, order the intervals of $B(M)$ as $ \{I_1^1, I_2^1, \ldots, I_n^1\}$ where $I_n^1$ is the unique infinite interval. For $1\leq i\leq \ell-1$, and $1\leq j\leq n-1$, let $I_j^i = (\phi_i\circ \cdots \circ \phi_1)(I_j^1)\in  B(M_{i+1})$, where $\phi_i$ is as in \cref{cor_matching}. The intervals $I_j^i$ are distinct for different $j$ by \cref{lem:matching}. 

\begin{figure}
\centering
\begin{tikzpicture}
\draw[-] (-4,-2.5) -- (-4,2.5);
\fill [orange!30] (2.2,-2.5) rectangle (3.8,2.5);

\draw[-] (-4.1,-1) -- (2, -1);
\draw[-] (-4.1,1) -- (2, 1);
\draw[-] (-4.1,2) -- (2, 2);
\draw[-] (-4.1,0) -- (2, 0);

\draw[black!20] (2,-1) -- (4, -1);
\draw[black!20] (2,1) -- (4, 1);
\draw[black!20] (2,2) -- (4, 2);
\draw[black!20] (2,0) -- (4, 0);

\draw[-] (4,-1) -- (6, -1);
\draw[-] (4,1) -- (6, 1);
\draw[-] (4,2) -- (6, 2);
\draw[-] (4,0) -- (6, 0);

\node[left] at (-4,-1) {$0$};
\node[left] at (-4,0) {$1$};
\node[left] at (-4,1) {$2$};
\node[left] at (-4,2) {$3$};

\node[below] at (-3,-1.3) {$M$};
\node[below] at (-2,-1.3) {$M_2$};
\node[below] at (-1,-1.3) {$M_3$};
\node[below] at (0,-1.3) {$M_4$};
\node[below] at (1,-1.3) {$M_5$};

\node[below] at (5,-1.3) {$M_{\ell-1}$};
\node[below] at (6,-1.3) {$N$};

\node[above] at (-3,-1) {$m_1$};
\node[above] at (6,2) {$m_\ell$};

\node[above] at (-3,1) {$v^1_j$};
\node[above] at (6,-1) {$v^\ell_j$};
\node [draw, shape = circle, fill = black, minimum size = 0.15cm, inner sep=0pt] (pt1) at (-3,-1){};
\node [draw, shape = circle, fill = black, minimum size = 0.15cm, inner sep=0pt] (pt2) at (-2,-0.5){};
\node [draw, shape = circle, fill = black, minimum size = 0.15cm, inner sep=0pt] (pt3) at (-1,0){};
\node [draw, shape = circle, fill = black, minimum size = 0.15cm, inner sep=0pt] (pt4) at (0,0.8){};
\node [draw, shape = circle, fill = black, minimum size = 0.15cm, inner sep=0pt] (pt5) at (1,0.3){};
\node [draw, shape = circle, fill = black, minimum size = 0.15cm, inner sep=0pt] (pt6) at (2,0.8){};
\node [draw, shape = circle, fill = black, minimum size = 0.15cm, inner sep=0pt] (pt7) at (4,2){};
\node [draw, shape = circle, fill = black, minimum size = 0.15cm, inner sep=0pt] (pt8) at (5,1.5){};
\node [draw, shape = circle, fill = black, minimum size = 0.15cm, inner sep=0pt] (pt9) at (6,2){};

\draw (pt1) -- (pt2) -- (pt3) -- (pt4) -- (pt5) -- (pt6);
\draw[dashed] (pt6) -- (2.3,1) -- (2.6, 1.5) -- (2.9, -0.5) -- (3.2, 1.0) -- (3.5, 0.5) -- (3.8, 1) --  (pt7);
\draw (pt7) -- (pt8) -- (pt9);

\node [draw, shape = rectangle, fill = black, minimum size = 0.15cm, inner sep=0pt] (sq1) at (-3,1){};
\node [draw, shape = rectangle, fill = black, minimum size = 0.15cm, inner sep=0pt] (sq2) at (-2,1.5){};
\node [draw, shape = rectangle, fill = black, minimum size = 0.15cm, inner sep=0pt] (sq3) at (-1,1.){};
\node [draw, shape = rectangle, fill = black, minimum size = 0.15cm, inner sep=0pt] (sq4) at (0,0){};
\node [draw, shape = rectangle, fill = black, minimum size = 0.15cm, inner sep=0pt] (sq5) at (1,-0.3){};
\node [draw, shape = rectangle, fill = black, minimum size = 0.15cm, inner sep=0pt] (sq6) at (2,-0.7){};
\node [draw, shape = rectangle, fill = black, minimum size = 0.15cm, inner sep=0pt] (sq7) at (4,-0.5){};
\node [draw, shape = rectangle, fill = black, minimum size = 0.15cm, inner sep=0pt] (sq8) at (5,-0.7){};
\node [draw, shape = rectangle, fill = black, minimum size = 0.15cm, inner sep=0pt] (sq9) at (6,0){};

\draw (sq1) -- (sq2) -- (sq3) -- (sq4) -- (sq5) -- (sq6);
\draw[line width=1mm, color=red] (pt3) -- (pt4);
\draw[line width=1mm, color=red] (pt5) -- (pt6);

\draw[dashed] (sq6) -- (2.3,-1.5) -- (2.6, -2) -- (2.9, -1.8) -- (3.2, -1) -- (3.5, -1.5) -- (3.8, -0.2) --  (sq7);
\draw (sq7) -- (sq8) -- (sq9);

\draw[line width=1mm, color=blue] (sq1) -- (sq2);
\draw[line width=1mm, color=blue] (sq3) -- (sq4);
\draw[line width=1mm, color=blue] (sq8) -- (sq9);
\end{tikzpicture}
\caption{The filled square depicts the evolution of the right endpoint $v_j^1$ of the $j$-th interval of $B(M)$ as it is matched through the sequence of merge trees. If $\delta v_j^i>0$, then the edge connecting $v_j^i$ and $v_j^{i+1}$ is drawn in thick blue. The merge point $m_i$ between the images of the generators $g$ and $h$ is depicted by a filled circle. If $\delta m_i>0$, the edge connecting $m_i$ and $m_{i+1}$ is drawn in thick red. }
\label{fig:track}
\end{figure}

Then, let $v^i = (v_1^i, \ldots, v_{n-1}^i)$, where $v_{j}^i$ denotes the right endpoint of $I_j^i$. Now, let 
\[\hat{v}_j^i = \begin{cases} |v_j^i-1| & \text{if $v_j^i\leq 2$} \\ |v_j^i-3| & \text{if $v_j^i\geq 2$}\end{cases}, \qquad\qquad \delta v_j^i = \max(\hat{v}_{j}^{i} - \hat{v}_j^{i+1},0).\]
In other words, $\delta v_j^i$ is non-negative only if the right endpoint is ``moving towards'' 1 on $(-\infty,2]$, and 
 towards 3 on $[2,\infty)$. 

And lastly, observe that since $\sigma\colon G_0 \to G_{\ell}$ from \cref{lem:sigma-comp} is a bijection, and $\Ss$ does not have a solution, \cref{lem:bijection-3part} tells us that there exists a pair of generators $g,h\in G_0$ with $\mrg_{P_1}(g,h)=0$, and $\mrg_{P_\ell'}(\sigma(g), \sigma(h)) = 3$. For $1\leq i \leq \ell$, let \[m_i = \mrg_{P_{i}'}((\sigma_{i-1}\circ \cdots \circ \sigma_1)(g),(\sigma_{i-1}\circ \cdots \circ \sigma_1)(h)),\] 
which gives $m_\ell = 3$; and for $1\leq i\leq \ell -1$, set $\delta m_i = m_{i+1}-m_i$ if $1\leq m_i\leq m_{i+1}\leq 2$, and $\delta m_i = 0$, otherwise. Note that $m_1 = \mrg_{P_1}(g,h)=0$ by \cref{eq:equalmerge}. 

\begin{lemma}
\label{lem:single-step-merge}
For all $1\leq i\leq \ell-1$, $d^p(P_i',P_{i+1}) \geq ||(\delta v_1^i, \ldots, \delta v_{n-1}^i, \delta m_i)||_p.$
\end{lemma}
\begin{proof}
    By \cref{cor_matching}, $|v_j^{i+1} - v_j^i| = |\gr_{i+1}(r_j^i) - \gr_i'(r_j^i)|$ for some relation $r_j^i$, and $r_j^i \neq r_{j'}^{i}$ for $j\neq j'$. 
    Therefore, 
    \[d^p(P_i',P_{i+1})\geq \left(\sum_{j=1}^{n-1} |\gr_{i+1}(r_j^i) - \gr_i'(r_j^i)|^p\right)^{1/p} =  ||v^{i+1} - v^{i}||_p.\]
    Since, by \cref{lem_refinement}, we have that $v_j^i\geq 2$ and $v_j^{i+1}\geq 2$, or $v_j^i\leq 2$ and $v_j^{i+1}\leq 2$, the following inequality also holds, $|v_j^{i+1} - v_j^i| \geq \delta v_j^i.$

    Note, however, that the relation $r$ from \cref{lem_long_branch_merge_levels} merging $(\sigma_i\circ \cdots \cdot \sigma_0)(g)$ and $(\sigma_i\circ \cdots \cdot \sigma_0)(h)$, may be equal to some $r_i^j$. This is not an issue, however, for if $r = r_j^i$, and $\delta m_i \neq 0$, then the right-endpoint of $I^i_j$ is moving \emph{away} from 1 on $[1,2]$ and thus $\delta v_j^i = 0$. It is also easy to see that, $|v_j^{i+1}-v^i_j| = \delta m_i$ if $r=r_j^i$ and $\delta m_i\neq 0$. Therefore, \[||v^i - v^{i+1}||_p \geq ||(\delta v_1^i, \ldots, \delta v_{n-1}^i, \delta m_i)||_p.\]
    This concludes the proof.
\end{proof}

\begin{theorem}
If $\Ss$ does not have a solution, then $d_I^p(M,N) \geq n^{\frac 1p}$.
\end{theorem}
\begin{proof}
Let $d_I^p(M,N) > \epsilon>0$, and let $M_i$, $P_i$, and $P_i'$ be as in \cref{lem_refinement}. Denote by $\phi$ the 
matching defined in \cref{lem:matching}. Then, since $\phi$ is a bijection by \cref{lem:matching}, each $\phi(I^1_j)$ must be matched to a interval in $B(N)$ whose right endpoint is either $1$ and $3$. Moreover, the sequence of presentations are, by virtue of \cref{lem_refinement} (iv), designed such that if $[a,b)\in B(M_i)$ is any interval with $b<z$ for $z\in \{0,1,2,3\}$, then $\phi_i([a,b)) = [c,d)$ where $d\leq z$ for the same $z$. Or, symmetrically: if, $b>z$, then $d\geq z$. 

Let $\Delta(I^1_j)= \sum_{i=1}^{\ell-1} \delta v_j^i$. If the right endpoint $v_j^1$ of $I^1_j$ is 0, then the sequence of points $v_j^1, \ldots, v_j^{\ell}$ must traverse the interval $0$ to $1$, and thus $\Delta(I^1_j)\geq 1$. Similarly, if $v_j^1=3$, then the corresponding sequence of points must traverse the interval from $3$ to $2$, and thus $\Delta(I^1_j)\geq 1$. Finally, if we let $\Delta(g,h) = \sum_{i=1}^{\ell-1} \delta m_i$, then $\Delta(g,h)\geq 1$ because the relation between $g$ and $h$ moves from $0$ to $3$, and thus traverses the interval from $1$ to $2$. From \cref{lem:single-step-merge}, we get 
\begin{linenomath}\begin{align*}
    d_I^p(M,N) + \epsilon > \sum_{i=1}^{\ell-1} d^p(P_i', P_{i+1}) &\geq \sum_{i=1}^{\ell-1} ||(\delta v_1^i, \ldots, \delta v_{n-1}^i, \delta m_i)||_p\\
    \text{(triangle inequality)}&\geq ||\Delta(I_1^1), \ldots, \Delta(I_{n-1}^1), \Delta(g,h)||_p \\
    & \geq ||(1,1, \ldots,1)||_p = n^{1/p}.
\end{align*}\end{linenomath}
Since this holds for every sufficiently small $\epsilon>0$, the lemma follows.
\end{proof}

\section{Computing $d_I^p$ in multiparameter persistence}
\label{sec:multi]}
In this section we prove that computing the $p$-presentation distance for 2-parameter persistence modules is NP-hard, and observe that the result extends to $t$-parameter modules for any $t\geq 2$. Several steps of the analysis are completely analogous to the proof for merge trees. However, working with presentations where the relations can be more complicated than simply identifying two points adds another level of complexity.

First, we fix some notation. Let $t\geq 1$ be an integer. We shall consider functors $M\colon \R^t\to \mathrm{Vec}$ where $\R^t = \R\times \dots \times \R$ has the product order, and $\mathrm{Vec}$ is the category of vector spaces over a finite field $\F$ that is fixed throughout the rest of the paper. We refer to such a functor as a \emph{$t$-parameter (persistence) module}. Recall that we write $M_q$ and $M_{q\to r}$ for $M$ applied to a point or a morphism $q\leq r$.
A morphism $f\colon M\to N$ is to be understood as a natural transformation between functors. 
That is, $f$ is a collection of morphisms $f_q\colon M_q\to N_q$ satisfying $f_r\circ M_{q\to r} = N_{q\to r}\circ f_q$ for all $q\leq r\in \R^t$.

\begin{remark}
Though we have not checked this carefully, we believe our results extend to infinite fields under the assumption that field elements can be manipulated in constant time. Since in persistence one usually works with finite fields, we assume $\F$ to be finite for convenience.
\end{remark}

\subsection{The $p$-presentation distances}
\label{subsec.mod.ppresdist}
Given a finite set $G$, let $\mathbb{F}^G$ denote the free $\mathbb{F}$-vector space generated by $G$. Explicitly, the elements of $\mathbb{F}^G$ are formal sums $\sum_{g\in G} \lambda_g g$, for $\lambda_g\in \mathbb{F}$, and addition and scalar multiplication are defined in the obvious way.
For a subset $U$ of $\F^G$, let $\langle U\rangle$ be the subspace of $\mathbb{F}^G$ spanned by the vectors in $U$.
A \emph{presentation} is a triple $P = (G,R,\gr)$ where $G$ is a finite multiset of \emph{generators}, $R$ is a finite multiset of elements of $\F^G$ called \emph{relations}, and $\gr\colon G\sqcup R\to \R^t$ satisfies $r\in \langle G_{\gr(r)} \rangle$ for every $r\in R$. Here $G_p$ (or $G_p^{\gr}$ if we want to specify the grade function) is $\{g\in G\mid \gr(g)\leq p\}$, and the multiset $R_p$ is defined in the analogous way.

Associated to any presentation, one has a persistence module $\M(P)$ defined pointwise by the quotient $\M(P)_p = \langle G_p\rangle /\langle R_p\rangle $ together with linear maps $\M(P)_{p\to p'}$ induced by the inclusions $G_p\hookrightarrow G_{p'}$ and $R_p\hookrightarrow R_{p'}$.
If $M\cong \M(P)$, we say that $P$ is a \emph{presentation of $M$}, and $M$ is \emph{finitely presented}.
If $P=(G,R,\gr)$ and $P'=(G',R',\gr')$ are presentations of two (not necessarily isomorphic) persistence modules, and there is an isomorphism $\sigma\colon \mathbb{F}^G \to \mathbb{F}^{G'}$ that restricts to bijections $G\to G'$ and $R\to R'$, then we say that $P$ and $P'$ are \emph{$\sigma$-compatible}, or just \emph{compatible}.
Equivalently, one can view $\sigma$ as a bijection $G\to G'$ whose induced isomorphism $\mathbb{F}^G \to \mathbb{F}^{G'}$ restricts to a bijection $R\to R'$.

\begin{remark}
Most persistence modules used in practice are finitely presented, since a presentation is a convenient way to describe a module for computational purposes.
But plenty of modules are not finitely presented: for instance, a $1$-parameter module supported on $(0,1]$ or a $2$-parameter module supported on the triangle with corners at $(1,0)$, $(0,1)$ and $(1,1)$ would both require an infinite number of generators to present.
In this paper, all modules are assumed to be finitely presented.
\end{remark}

Similarly to our proof for merge trees, a key idea is that modules often have compatible presentations despite looking quite different.
The following is an example of this phenomenon.
\begin{example}
\label{ex_modules_presentations}
Let $P = (G,R,\gr)$ be the presentation with $G=\{g\}$, $R=\{g\}$, $\gr(g) = (0,0)$ and $\gr(g) = (2,0)$ (for the relation $g$).
Let $P' = (G',R',\gr')$ be the presentation with $G' = \{g_1,g_2\}$, $R = \{g_1,g_1-g_2\}$, and $\gr(g_1) = (0,1)$, $\gr(g_2) = (1,0)$, $\gr(g_1) = (0,3)$ (for the relation $g_1$) and $\gr(g_1-g_2) = (1,2)$.
Then $\M(P)$ and $\M(P')$ are the $2$-parameter modules illustrated in \cref{fig_module_presentation}.
Like for the merge trees presentations in \cref{ex.comppres}, we can express presentations as matrices where the rows correspond to generators and the columns to relations,
\[
P=
\begin{bNiceMatrix}[
  first-row,code-for-first-row=\scriptstyle,
  first-col,code-for-first-col=\scriptstyle,
]
& (2,0) \\
(0,0) & 1
\end{bNiceMatrix}\qquad \qquad P' = \begin{bNiceMatrix}[
  first-row,code-for-first-row=\scriptstyle,
  first-col,code-for-first-col=\scriptstyle,
]
& (0,3) & (1,2)  \\
(0,1) & 1 & 1 \\
(1,0) & 0 & -1
\end{bNiceMatrix}.
\]

These presentations are clearly not compatible. On the other hand, a presentation of $\M(P)$ compatible with $P'$ is given by
\[
\tilde P = \begin{bNiceMatrix}[
  first-row,code-for-first-row=\scriptstyle,
  first-col,code-for-first-col=\scriptstyle,
]
& (1,0) & (0,0)  \\
(0,0) & 1 & 1 \\
(0,0) & 0 & -1
\end{bNiceMatrix}.
\]
Note how the generators in $\tilde P$ are immediately identified by the relation $g_1-g_2$ (where $g_i$ is the generator of row $i$) in $\tilde P$, so the two generators and $g_1-g_2$ together have the same effect as the single generator $g$ in $P$; see \cref{fig_module_presentation}.
\end{example}

\begin{figure}
\centering
\begin{tikzpicture}[scale=.9]
\foreach \x in {-1,4.5,10}{
\begin{scope}[xshift=\x cm, yshift=-1cm]
\draw[thick,<->] (0,5) to (0,0) to (5,0);
\draw[thick] (-.06,1) to (.06,1);
\node[left] at (0,1){$0$};
\draw[thick] (-.06,2) to (.06,2);
\node[left] at (0,2){$1$};
\draw[thick] (-.06,3) to (.06,3);
\node[left] at (0,3){$2$};
\draw[thick] (-.06,4) to (.06,4);
\node[left] at (0,4){$3$};
\draw[thick] (1,-.06) to (1,.06);
\node[below] at (1,0){$0$};
\draw[thick] (2,-.06) to (2,.06);
\node[below] at (2,0){$1$};
\draw[thick] (3,-.06) to (3,.06);
\node[below] at (3,0){$2$};
\draw[thick] (4,-.06) to (4,.06);
\node[below] at (4,0){$3$};
\end{scope}
}
\begin{scope}
\fill[red, opacity=0.3] (0,0) rectangle (2,3.5);
\node[below] at (0,0){$g$};
\draw[color=black,fill=black] (0,0) circle (.04);
\node[below] at (2,0){$g$};
\draw[color=blue,fill=blue] (2,0) circle (.04);
\node at (1,1.5){$\langle g \rangle \cong \F$};
\node at (3,1.5){$\langle g \rangle/\langle g \rangle \cong 0$};
\node at (1.5,-2){$\M(P)$};
\filldraw[color=white,path fading = south] (0,3.1) rectangle (3,3.5);
\end{scope}
\begin{scope}[xshift=5.5cm]
\fill[red, opacity=0.3] (3.5,0) to (1,0) to (1,1) to (0,1) to (0,3) to (3.5,3);
\fill[red, opacity=0.3] (1,1) rectangle (3.5,2);
\draw[dotted] (1,2) to (1,3);
\coordinate (g1) at (0,1);
\coordinate (g2) at (1,0);
\coordinate (1g1) at (0,3);
\coordinate (g1g2) at (1,2);
\coordinate (a) at (.5,1.5);
\coordinate (b) at (2,.5);
\coordinate (c) at (2,1.5);
\coordinate (d) at (2,4);
\node[below] at (g1){$g_1$};
\draw[color=black,fill=black] (g1) circle (.04);
\node[below] at (g2){$g_2$};
\draw[color=black,fill=black] (g2) circle (.04);
\node[above] at (1g1){$g_1$};
\draw[color=blue,fill=blue] (1g1) circle (.04);
\node[above] at (g1g2){$g_1-g_2$};
\draw[color=blue,fill=blue] (g1g2) circle (.04);
\node at (a){$\langle g_2 \rangle$};
\node at (b){$\langle g_1 \rangle$};
\node at (c){$\langle g_1,g_2 \rangle$};
\node at (d){$\langle g_1,g_2 \rangle/\langle g_1-g_2 \rangle$};
\draw (2,2.5) to (2,3.7);
\node at (1.5,-2){$\M(P')$};
\filldraw[color=white,path fading = west] (3.1,0) rectangle (3.5,3);
\end{scope}
\begin{scope}[xshift=11cm]
\fill[red, opacity=0.3] (0,0) rectangle (2,3.5);
\node[below] at (0,0){$g_1,g_2$};
\node at (0,-.6){$g_1-g_2$};
\draw[color=black,fill=black] (-.03,.03) circle (.04);
\draw[color=blue,fill=blue] (.03,-.03) circle (.04);
\node[below] at (2,0){$g_1$};
\draw[color=blue,fill=blue] (2,0) circle (.04);
\node at (1,2){$\langle g_1,g_2 \rangle/$};
\node at (1,1.5){$\langle g_1-g_2 \rangle$};
\node at (1,1){$\cong \F$};
\node at (2,4.5){$\langle g_1,g_2 \rangle/\langle g_1,g_1-g_2 \rangle \cong 0$};
\draw (2.2,4.2) to (3,1.5);
\node at (1.5,-2){$\M(\tilde P)$};
\filldraw[color=white,path fading = south] (0,3.1) rectangle (2,3.5);
\end{scope}
\end{tikzpicture}
\caption{The modules in \cref{ex_modules_presentations}.
The generators are marked with black dots, the relations with blue dots, and the shade of red shows the dimension of the vector space.}
\label{fig_module_presentation}
\end{figure}

For $\sigma$-compatible presentations $P=(G,R,\gr)$ and $P'=(G',R',\gr')$, and $p\in [1, \infty]$, let
\[
d^p(P,P',\sigma) = \begin{cases} \left(\sum_{h\in G\cup R} \|\gr(h)-\gr'(\sigma(h))\|_p^p\right)^{\frac 1 p} & \text{if $p<\infty$} \\
\max_{h\in G\cup R} \max_{|\cdot |}(\gr(h)-\gr'(\sigma(h))) & \text{if $p=\infty$},\end{cases}
\]
where $\max_{|\cdot |}(a,b) = \max\{|a|,|b|\}$.
If $\sigma$ is the identity, we shall write $d^p(P,P')$ unless otherwise stated.
For $t\geq 1$ and $t$-parameter finitely presented modules $M$ and $N$, 
\[
\hat d_I^p(M,N) = \inf\{d^p(P,P',\sigma) \mid P\text{ and }P' \text{ are $\sigma$-compatible presentations of }M\text{ and } N\}.\]

Just like for merge trees, $\hat d_I^p$ does not satisfy the triangle inequality for modules for $p<\infty$.
An example showing this is given in \cite[Example~3.1]{bjerkevik2021ell}, very similar in spirit to \cite[Example 4.2]{cardona2022universal}, which we paraphrased in \cref{ex_triangle_fail}.
Again, this motivates a definition of $d^p_I$ that ``fixes'' the triangle inequality:

\begin{definition}
\label{def_pres_dist_mods}
Let $t\geq 1$. The $p$-presentation distance between finitely presented $t$-parameter modules $M$ and $N$ is 
\[d_I^p(M,N) = \inf\sum_{i=1}^{\ell-1} \hat{d}_I^p(M_i,M_{i+1}),\] where the infimum is taken over all sequences $M=M_1, M_2,\dots,M_\ell=N$ of modules.
\end{definition}
It is shown in \cite[Theorem~1.7]{bjerkevik2021ell} that $d^\infty_I$ is equal to the interleaving distance.

\subsection{Reducing from constrained invertibility problems to modules}

Constrained invertibility problems, or CI problems for short, were first introduced in \cite{bjerkevik2017computational} and subsequently shown to be NP-hard to solve in \cite{bjerkevik2020computing}.
\begin{definition}
A constrained invertibility (CI) problem (over the finite field $\mathbb{F}$) is a pair $(\Pp,\Q)$ of $n\times n$-matrices where each entry of the matrices is either $0$ or $*$.
A solution of a CI problem is a pair $(A,B)$ of $n\times n$-matrices with entries in $\F$ such that $B = A^{-1}$, and if $\Pp_{i,j} = 0$, then $A_{i,j} = 0$, and if $\Q_{i,j} = 0$, then $B_{i,j} = 0$.
\end{definition}

\begin{example}
The following shows a CI problem and a solution of it. \[
\Pp = \begin{bmatrix} * & * & * \\ * & 0 & * \\ * & * & 0\end{bmatrix}, \Q = \begin{bmatrix} * & * & * \\ * & * & 0 \\ * & 0 & *\end{bmatrix}\qquad\rightarrow\qquad \begin{bmatrix} 1 & 1 & 1 \\ 1 & 0 & 1 \\ 1 & 1 & 0\end{bmatrix}\begin{bmatrix} -1 & 1 & 1 \\ 1 & -1 & 0 \\ 1 & 0 & -1\end{bmatrix} = \begin{bmatrix} 1 & 0 & 0 \\0 & 1 & 0 \\ 0 & 0 & 1 \end{bmatrix} \]
\end{example}

\begin{example}
The following instance of a CI problem has no solution, because the (1, 1) entry of the product is always 0, no matter what values are chosen. \[
\Pp=\begin{bmatrix} 0 & * & 0 \\ * & * & * \\ * & * & *\end{bmatrix}\qquad\qquad \Q = \begin{bmatrix} * & * & * \\ 0 & * & * \\ * & * & *\end{bmatrix}.\]
\end{example}

\begin{theorem}[\cite{bjerkevik2020computing}]
Deciding if a CI problem has a solution is $NP$-hard.
\end{theorem}
We will reduce deciding if a CI problem has a solution to computing $d_I^p$ for 2-parameter persistence modules. 
\subsubsection{The $2$-parameter modules $M$ and $N$}
\label{sec:pmodMN}
We now describe how to associate two $2$-parameter modules to an instance $(\Pp, \Q)$ of a CI problem with $n\geq 2$. In the following, $p\in [1, \infty)$, $K$ denotes the total number of $0$-entries in $\Pp$ and $\Q$, and $C=4(Kn+1)^{1/p}$. Associated to the 0-entries in $\Pp$ and $\Q$ we select $K$ distinct points $q_1, \ldots, q_K$ in $\R^2$ such that for $k\neq k'$, we have $q_k+(3,0) \leq (0,0)$ and $q_k+(4,4)\ngeq q_{k'}$. Moreover, if $q_k$ corresponds to a 0-entry in $\Pp$, then the $x$-coordinate of $q_k$ is an even integer, while if it corresponds to a 0-entry in $\Q$, then its $x$-coordinate is an odd integer.
For instance, we can choose each $q_k$ to be of the form $(5i-5K-3,-5i-5K)$ for $-K\leq i \leq K$ with $i$ of the appropriate parity.

\begin{figure}
    \centering

\begin{tikzpicture}[line cap=round,line join=round,x=1.0cm,y=1.0cm,scale=0.7]
\clip (-13,-8) rectangle (6,6);

    \node at (-5,4) {$\Q = \begin{bmatrix} * & * & * \\ * & * & * \\ * & \boxed{0} & * \end{bmatrix}$};
    \node at (-9,4) {$\Pp= \begin{bmatrix} * & \boxed{0} & * \\ \boxed{0} & * & * \\ * & * & * \end{bmatrix}$};

      \node [draw, shape = circle, fill = black, minimum size = 0.1cm, inner sep=0pt] (c) at (4,4){};
    \node[right] at (4,4) {$(C,C)$};
        \node[above] (Rn) at (4,5) {$R_N$};
    \node[left] (Rm) at (3,4.5) {$R_M$};
    \node[left] (ginf) at (3,3.5) {$g_{\infty,i}$};
        \node[below] (hinf) at (4,3) {$h_{\infty, i}$};
        \draw[dashed] (Rn) -- (c);
            \draw[dashed] (Rm) -- (c);
                \draw[dashed] (hinf) -- (c);
                    \draw[dashed] (ginf) -- (c);
                    
    \draw (-10,1) -- (5,1);
    \draw (0,-7) -- (0,5);
    \node[left] (P1) at (-8, -1) {$q_1$};

    \node[above] (hi1) at (-9,-0.6) {$h_{1,2}, h_{1,3}$};
     \node[above] (gi1) at (-6.5,-0.6) {$g_{1,1}, g_{1,3}$};
    \node[right] at (-7, -1) {$q_1 + (3,0)$};
    \node[left] (g12) at (-8, -1.5) {$g_{1,2}$};
     \node[right] (h11) at (-7.25, -1.5) {$h_{1,1}$};
    \node [draw, shape = circle, fill = black, minimum size = 0.1cm, inner sep=0pt] (pt1) at (-8,-1){};
    \node [draw, shape = circle, fill = black, minimum size = 0.1cm, inner sep=0pt] (pt1) at (-7.25,-1){};
    \node [draw, shape = circle, fill = black, minimum size = 0.1cm, inner sep=0pt] (pt1) at (-7.5,-1){};
    \node [draw, shape = circle, fill = black, minimum size = 0.1cm, inner sep=0pt] (pt1) at (-7.75,-1){};
    \draw[dashed] (hi1) -- (-7.75,-1);
    \draw[dashed] (gi1) -- (-7.5,-1);
    \draw[dashed] (g12) -- (-8,-1);
    \draw[dashed] (h11) -- (-7.25,-1);
    \begin{scope}[shift={(2,-2)}]

    \node[left] (P2) at (-8, -1) {$q_2$};
   \node[above] (hi1) at (-9,-0.6) {$h_{2,1}, {h_{2,3}}$};
    \node[above] (gi1) at (-7,-0.6) {$g_{2,2}, g_{2,3}$};
    \node[right] at (-7, -1) {$q_2 + (3,0)$};
    \node[left] (g12) at (-8, -1.5) {$g_{2,1}$};
     \node[right] (h11) at (-7.25, -1.5) {$h_{2,2}$};
    \node [draw, shape = circle, fill = black, minimum size = 0.1cm, inner sep=0pt] (pt1) at (-8,-1){};
    \node [draw, shape = circle, fill = black, minimum size = 0.1cm, inner sep=0pt] (pt1) at (-7.25,-1){};
    \node [draw, shape = circle, fill = black, minimum size = 0.1cm, inner sep=0pt] (pt1) at (-7.5,-1){};
    \node [draw, shape = circle, fill = black, minimum size = 0.1cm, inner sep=0pt] (pt1) at (-7.75,-1){};
   \draw[dashed] (hi1) -- (-7.75,-1);
    \draw[dashed] (gi1) -- (-7.5,-1);
    \draw[dashed] (g12) -- (-8,-1);
    \draw[dashed] (h11) -- (-7.25,-1);
    \end{scope}

    \begin{scope}[shift={(4,-4)}]

    \node[left] (P3) at (-8, -1) {$q_3$};
    \node[above] (hi1) at (-9,-0.6) {$g_{3,1}, g_{3,2}$};
     \node[above] (gi1) at (-7,-0.6) {$h_{3,1}, h_{3,3}$};
    \node[right]  at (-7, -1) {$q_3 + (3,0)$};
    \node[left] (g12) at (-8, -1.5) {$h_{3,2}$};
     \node[right] (h11) at (-7.25, -1.5) {$g_{3,3}$};
    \node [draw, shape = circle, fill = black, minimum size = 0.1cm, inner sep=0pt] (pt1) at (-8,-1){};
    \node [draw, shape = circle, fill = black, minimum size = 0.1cm, inner sep=0pt] (pt1) at (-7.25,-1){};
    \node [draw, shape = circle, fill = black, minimum size = 0.1cm, inner sep=0pt] (pt1) at (-7.5,-1){};
    \node [draw, shape = circle, fill = black, minimum size = 0.1cm, inner sep=0pt] (pt1) at (-7.75,-1){};
   \draw[dashed] (hi1) -- (-7.75,-1);
    \draw[dashed] (gi1) -- (-7.5,-1);
    \draw[dashed] (g12) -- (-8,-1);
    \draw[dashed] (h11) -- (-7.25,-1);
    \end{scope}
        \node (eh) at (-8.7, 4.8) {};
    \node (eh2) at (-9.1, 3.7) {};
    \node (eh3) at (-4, 3.55) {};
\draw[<->, dashed, thick] (eh) [out=150, in=170, distance=6cm] to (P2);
\draw[<->, dashed, thick] (eh2) [out=-40, in=45, distance=3cm] to (P1);
\draw[<->, dashed, thick] (eh3) [out=20, in=-45, distance=10cm] to (P3);

\end{tikzpicture}
\caption{Depicted are the grades of the generators and relations of the persistence modules $M$ and $N$ associated to the CI problem given in the top left quadrant; see \cref{sec:pmodMN}.}
\label{fig:CIMod}
\end{figure}

We define presentations $P_M=(G_M,R_M,\gr_M)$ and $P_N=(G_N,R_N,\gr_N)$, illustrated in \cref{fig:CIMod}, and let $M=\M(P_M)$ and $N=\M(P_N)$.
The generators and relations of the presentations are
\begin{align*}
G_M &= \{g_{k,i}\mid k\in \{1,\dots,K,\infty\}, 1\leq i\leq n\} & R_M &=\{r_{k,i}\mid 1\leq i\leq n, 1\leq k\leq K\}\\
G_N &= \{h_{k,i}\mid k\in \{1,\dots,K,\infty\}, 1\leq i\leq n\} & R_N &=\{s_{k,i}\mid 1\leq i\leq n, 1\leq k\leq K\}
\end{align*}
where $\gr_{M}(r) = (C,C)$ for all $r\in R_M$, $\gr_{N}(r) = (C,C)$ for all $r\in R_N$, and $r_{k,i}=g_{k,i} - g_{\infty,i}$ and $s_{k,i} = h_{k,i} - h_{\infty,i}$.
(It follows from this that both $M$ and $N$ decompose as direct sums of $n$ submodules, where the $i^\text{th}$ submodule of $M$ is generated by all generators of the form $g_{k,i}$, and similarly for $N$.)
It remains to specify the grades of the generators. To ease the notation, we shall suppress the subscripts of $\gr_M$ and $\gr_N$ when it is clear from context which function is used. 

\noindent \textbf{Grades of generators:} If $\Pp$ has a zero in position $(i,j)$ associated to $q_k$, we let
\begin{align*}
\gr(g_{k,j}) &= q_k,& \gr(g_{k,j'}) &= q_k + (2,0), 1\leq j'\leq n, j'\neq j,\\
\gr(h_{k,i}) &= q_k + (3,0),&\gr(h_{k,i'}) &= q_k + (1,0), 1\leq i'\leq n, i'\neq i.
\end{align*}
If $\Q$ has a zero in position $(i,j)$ associated to $q_k$, we let
\begin{align*}
\gr(g_{k,i}) &= q_k + (3,0),&\gr(g_{k,i'}) &= q_k + (1,0), 1\leq i'\leq n, i'\neq i,\\
\gr(h_{k,j}) &= q_k,& \gr(h_{k,j'}) &= q_k + (2,0), 1\leq j'\leq n, j'\neq j.
\end{align*}
Finally, let $\gr(g_{\infty,i}) = \gr(h_{\infty,i}) = (C,C)$ for all $i$.

We will prove the following theorem.
\begin{theorem}
\label{thm_mods}
Let $p\in [1, \infty)$. If $(\Pp,\Q)$ has a solution, then $d_I^p(M,N)\leq (Kn)^{\frac 1p}$. Otherwise, $d_I^p(M,N)\geq (Kn+1)^{\frac 1p}$.
\end{theorem}

It follows that (the decision version of) computing $d_I^p$ is NP-hard. 
\begin{corollary}
\label{cor_main_2_par}
Let $p\in [1,\infty)$. Given finitely presented $2$-parameter modules $M$ and $N$ and  $c>0$, it is NP-hard to decide if $d_I^p(M,N)\leq c$.
\end{corollary}

\subsubsection{Generalizing to $t \geq 2$ parameters}

It is not hard to show that \cref{cor_main_2_par} generalizes to an arbitrary number of parameters. To see this, let $M$ be a finitely presented 2-parameter module, let $\pi\colon \R^t \to \R^2$ denote the projection onto the first two coordinates, let $\iota\colon\R^2\to \R^t$ denote the inclusion into the first two coordinates, and for $J\subseteq \R^t$, let $\langle J \rangle = \{q'\in \R^t\mid \exists q\in J\text{ such that } q'\geq q \}$. 

For a $2$-parameter module $M$, define the $t$-parameter module $\widehat M$ by $\widehat M_q=M_{\pi(q)}$ for $q\in \langle \iota(\R^2)\rangle$, and $\widehat M_q=0$ otherwise.
For a presentation $P=(G,R,\gr)$ of a $t$-parameter module, define the $2$-parameter presentation $\bar P=(G,R,\pi\circ\gr)$.
Since $\pi$ preserves the relation $\leq$, $\bar P$ is a valid presentation.
If $P$ is a presentation of $\widehat M$, then $\bar P$ is a presentation of $M$.
\begin{lemma}
For any finitely presented 2-parameter modules $M$ and $N$, $d^p_I(M,N) = d^p_I(\widehat M,\widehat N)$.
\end{lemma}
\begin{proof}
The inequality $\hat d^p_I(M',N') \geq\hat d^p_I(\widehat{M'},\widehat{N'})$ for $2$-parameter modules $M'$ and $N'$ is immediate, as any presentations of $M'$ and $N'$ can be viewed as presentations of $\widehat{M'}$ and $\widehat{N'}$ by including the grades into $\R^t$.
It follows that $d^p_I(M,N) \geq d^p_I(\widehat M,\widehat N)$.
Conversely, let $P$ and $Q$ be compatible presentations of $t$-parameter modules $M'$ and $N'$.
Since $\pi$ does not increase distances, we have $d(\bar P,\bar Q)\leq d(P,Q)$.
Thus, $\hat d^p_I(\M(\bar P),\M(\bar Q))\leq \hat d^p_I(M',N')$.
By replacing $M_i$ (presented by $P_i$) with $2$-parameter modules $\M(\bar P_i)$ in \cref{def_pres_dist_mods}, we get $d^p_I(M,N) \leq d^p_I(\widehat M,\widehat N)$.
\end{proof}

The following is now immediate.
\begin{corollary}
\label{cor_main_t_par}
Let $p\in [1,\infty)$ and $t\geq 2$. Given finitely presented $t$-parameter modules $M$ and $N$ and $c>0$, it is NP-hard to decide if $d_I^p(M,N)\leq c$.
\end{corollary}

\subsection{High-level approach and parallels to the merge tree proof}
Given a CI problem $(\Pp,\Q)$ of $n\times n$-matrices, we have constructed modules $M$ and $N$ such that for $p\geq (C,C)\gg (0,0)$, $M_p\cong N_p \cong \F^n$.
Realizing a reasonably small $p$-presentation distance between $M$ and $N$ involves picking an isomorphism $\sigma\colon M_p\to N_p$ (which is analogous to the bijection $\sigma$ in \cref{lem:sigma-comp}), which we can write as an $n\times n$-matrix $A_\sigma$.
We interpret $(A_\sigma,A_\sigma^{-1})$ as an attempt at a solution of $(\Pp,\Q)$.
If $(A_\sigma,A_\sigma^{-1})$ is indeed a solution, then it turns out that we can choose compatible presentations of $M$ and $N$ such that each generator needs to be moved a distance of at most $1$ between the modules.
But if $(A_\sigma)_{i,j} \neq 0$ and $\Pp_{i,j} = 0$ (or similarly for $A_\sigma^{-1}$ and $\Q$) for some $i$ and $j$, we are forced to move a generator a distance of $3$ from $q_k$ to $q_k+(3,0)$, where $q_k$ is associated to $(i,j)$, giving a larger cost.
This is analogous to the extra cost $\Delta(g,h)$ of the relation moving from $0$ to $3$ in the proof for merge trees, and the latter parts of the two proofs share the same set of ideas. To obtain $d_I^p(M,N) \leq (Kn)^{\frac 1p}$, we need to avoid any extra cost, which we show is equivalent to solving the CI problem.

\subsection{Proof of \cref{thm_mods}}

In \cref{subsec_step_1_mods}, we will show that if $(\Pp,\Q)$ has a solution, then $d_I^p(M,N)\leq (Kn)^{\frac 1p}$; and in \cref{subsec_step_2_mods}, we will show that if $(\Pp,\Q)$ does not have a solution, then $d_I^p(M,N)\geq (Kn+1)^{\frac 1p}$.
From this, \cref{thm_mods} follows. 

\subsubsection{Some additional notation}
\label{subsec:addnot}

Let $P=(G,R,\gr)$ be a presentation of a 2-parameter persistence module and set $A=\M(P)$.
We define $A^\infty = \mathbb{F}^G/\langle R\rangle$.
Note that $A^\infty = A_q$ for any $q$ with $q\geq q'$ for all $q'\in \img(\gr)$.
If $P'=(G',R',\gr')$ is another presentation with $B=\M(P')$, then any morphism $f\colon A\to B$ restricts to a morphism $f_\infty\colon A^\infty\to B^\infty$ by choosing $f_\infty = f_q$ for any sufficiently large $q$.
For any $q\in \R^2$, let 
\[A_{q\to \infty}=A_{q\to q'}\colon A_q\to A^\infty,\]
where $q'\geq q$ is chosen such that $A_{q'}=A^\infty$.

We set $M^k = M_{q_k+(3,0)}$ and $N^k = N_{q_k+(3,0)}$ for $1\leq k\leq K$. Also note that, by construction, $M^\infty = M_{(C,C)}$ and $N^\infty = N_{(C,C)}$.
For each $1\le k\leq K$, $M^k$ has an ordered basis $(g_{k,1},\dots,g_{k,n})$, and $N^k$ has an ordered basis $(h_{k,1},\dots,h_{k,n})$:
By construction, $g_{k,i}$ has a grade of $q_k+(j,0)$ for some $j\leq 3$, and for $k'\neq k$, we have $\gr_M(g_{k',i})\leq q_{k'}+(4,4)\ngeq q_k$, so $g_{k',i}$ does not appear in $M^k$.
Similar observations hold for $h_{k,i}$.
The vector spaces ``at infinity'' have ordered bases defined as follows: 
$M^\infty$ has an ordered basis $(\bar g_{\infty,1},\dots,\bar g_{\infty,n})$, where by $\bar g_{\infty,i,}$, we mean the equivalence class $[g_{\infty,i}]\in \F^{G_M}/\langle R_M\rangle$.
(By construction, $[g_{\infty,i}]$ is the set of linear combinations of the elements in $\{g_{k,i}\}_{k=1}^K\cup \{g_{\infty,i}\}$.)
Note that there is an isomorphism $\tilde M^\infty \coloneqq \langle \{g_{\infty,1},\dots, g_{\infty,n}\}\rangle \to M^\infty$ given by $g_{\infty,i} \mapsto \bar g_{\infty,i}$ for all $i$.
Likewise, $N^\infty$ has an ordered basis $(\bar h_{\infty,1},\dots,\bar h_{\infty,n})$ which is canonically isomorphic to a module $\tilde N$ defined similarly to $\tilde M$.
We refer to these as the \emph{canonical bases} of $M^k$ and $N^k$ for $k\in \{1,\dots,K,\infty\}$.
Observe that \[M^{k\to \infty}\coloneqq M_{q_k+(3,0)\to (C,C)} \colon M^k\to M^\infty\] is the isomorphism given by $g_{k,i} \mapsto \bar g_{\infty,i}$ for all $i$, and the linear transformation is given by the identity matrix with respect to the canonical bases. Define $N^{k\to \infty}\colon N^k\to N^\infty$ accordingly. 

\subsubsection{Step 1: If $(\Pp,\Q)$ has a solution, then $d_I^p(M,N)\leq (Kn)^{\frac 1p}$.}
\label{subsec_step_1_mods}

Let $P=(G,R,\gr)$ and $P'=(G',R',\gr')$ be presentations, and suppose we have an isomorphism $\iota\colon \F^{G'}\to \F^G$.
We interpret this as a change of basis, where we replace the basis $G$ with a basis $\iota(G')$.
We call $P'$ a \emph{regeneration} of $P$ if, for all $q\in \R^2$, $\langle G_q\rangle = \langle \iota(G'_q)\rangle$ and $\langle R_q\rangle = \langle \iota(R'_q)\rangle$.

\begin{lemma}
\label{lem_reg_iso}
If $P'$ is a regeneration of $P$, then $\M(P') \cong \M(P)$.
\end{lemma}

\begin{proof}
We will show that $\iota$ as described above induces an isomorphism.
Recall that $\M(P)_q = \langle G_q\rangle /\langle R_q\rangle$ and $\M(P')_q = \langle G'_q\rangle /\langle R'_q\rangle$.
Since
\begin{linenomath}\begin{align*}
\iota(\langle G'_q\rangle) &= \langle \iota(G'_q)\rangle = \langle G_q\rangle \text{ and}\\
\iota(\langle R'_q\rangle) &= \langle \iota(R'_q)\rangle = \langle R_q\rangle,
\end{align*}\end{linenomath}
$\iota$ induces an isomorphism $\bar \iota_q\colon \M(P)_q\to \M(P')_q$.
To see that $\bar \iota = \{\bar \iota_q\}_{q\in \R^2}$ is indeed a morphism of persistence modules (and thus, an isomorphism), observe that for $q\leq q'$, both $\bar \iota_{q'}\circ \M(P)_{q\to q'}$ and $\M(P')_{q\to q'}\circ \bar \iota_q$ send the equivalence class of $v\in \F^{G'}$ to the equivalence class of $\iota(v) \in \F^G$ in the relevant quotient vector spaces.
\end{proof}

Now suppose $(\Pp,\Q)$ has a solution $(S,T)$.
We will construct compatible presentations $P'_M$ and $P'_N$ and show that $M\cong \M(P'_M)$ and $N\cong \M(P'_N)$, respectively (\cref{lem_pres_of_M_and_N}), and  that $d^p(P'_M,P'_N,\sigma) = (Kn)^{\frac 1p}$ for some $\sigma$ (\cref{lem_comp_dp}).
From this, it follows that $d_I^p(M,N)\leq (Kn)^{\frac 1p}$, which is our goal.

With respect to the canonical bases given in \cref{subsec:addnot}, let $\phi_k\colon M^k \to N^k$ be the isomorphism with transformation matrix $S$ for $k\in \{1,\dots,K,\infty\}$.
We have $\phi_\infty\circ M^{k\to \infty} = N^{k\to \infty}\circ \phi_k$ for $k\neq \infty$, since the transformation matrices of $M^{k\to \infty}$ and $N^{k\to \infty}$ with respect to the canonical bases are identity matrices.
Observe that $T=S^{-1}$ is the transformation matrix of $\phi_k^{-1}$ for all $k$.
Since $\F^{G_M} = \tilde M^\infty \oplus \bigoplus_{k=1}^K M^k$ and $\F^{G_N} = \tilde N^\infty \oplus \bigoplus_{k=1}^K N^k$, assembling the $\phi_k$ gives us an isomorphism $\phi\colon \F^{G_M}\to \F^{G_N}$.

Next we define $P'_M = (G'_M,R'_M,\gr'_M)$ and $P'_N = (G'_N,R'_N,\gr'_N)$ and show that these are regenerations of $P_M$ and $P_N$, respectively.

Let
\begin{align*}
G'_M &= \{g'_{k,i}\mid k\in \{1,\dots,K,\infty\}, i\in \{1,\dots,n\}\}\sse \F^{G_M},\\
G'_N &= \{h'_{k,i}\mid k\in \{1,\dots,K,\infty\}, i\in \{1,\dots,n\}\}\sse \F^{G_N},
\end{align*}
where the $g'_{k,i}$ and $h'_{k,i}$ will be specified below.
Now we define the isomorphisms $\iota_M\colon \F^{G'_M}\to \F^{G_M}$ and $\iota_N\colon \F^{G'_N}\to \F^{G_N}$ by defining how they act on the bases $G'_M$ and $G'_N$, respectively.
The purpose of the construction that follows is to make sure that $\phi$ induces bijections $G'_M \to G'_N$ and $R'_M \to R'_N$ of generators and relations, and that the grades of generators/relations that are matched do not differ by more than one in the $x$-coordinate and zero in the $y$-coordinate, which is needed in \cref{lem_comp_dp}.

Let $1\leq k\leq K$, and suppose $q_k$ corresponds to a zero in position $(a,b)$ in $\Pp$.
Choose ordered bases $(\iota_M(g'_{k,1}),\dots,\iota_M(g'_{k,n}))$ and $(\iota_N(h'_{k,1}),\dots,\iota_N(h'_{k,n}))$ of $M^k$ and $N^k$, respectively, such that
\begin{itemize}
\item[(i)] $\iota_M(g'_{k,1}) = g_{k,b}$,
\item[(ii)] $\iota_N(h'_{k,i}) \in N_{q_k+(1,0)}$ for all $i<n$,
\item[(iii)] $\phi_k(\iota_M(g'_{k,i})) = \iota_N(h'_{k,i})$ for all $1\leq i\leq n$.
\end{itemize}
Recall that by construction, $M_{q_k}$ has a basis $\{g_{k,b}\}$, and $M_{q_k+(2,0)} = M^k$ with basis $\{g_{k,1},\dots,g_{k,n}\}$; and $N_{q_k+(1,0)}$ has a basis $\{h_{k,1},\dots,h_{k,a-1},h_{k,a+1},\dots,h_{k,n}\}$, and $N_{q_k+(3,0)} = N^k$ with basis $\{h_{k,1},\dots,h_{k,n}\}$.
If $\phi_k(g_{k,b})\in N_{q_k+(1,0)}$, choosing elements satisfying (i)-(iii) is easy: Pick a basis $\{\iota_N(h'_{k,1}), \dots, \iota_N(h'_{k,n-1})\}$ of $N_{q_k+(1,0)}$ with $\iota_N(h'_{k,1}) = \phi(g_{k,b})$, let $\iota_N(h'_{k,n}) = h_{k,a}$ (which is not in $N_{q_k+(1,0)}$ by construction), and let $\iota_M(g'_{k,i}) = \phi^{-1}(\iota_N(h'_{k,i}))$ for all $i$.
Since $\Pp$ has a zero in position $(a,b)$, $A$ has a zero in position $(a,b)$, which means that
\[
\phi_k(g_{k,b})\in \langle \{h_{k,1},\dots,h_{k,a-1},h_{k,a+1},\dots,h_{k,n}\} \rangle = N_{q_k+(1,0)},
\]
so we can indeed satisfy (i)-(iii).
Let $\gr_M'(g'_{k,1}) = q_k$, $\gr_M'(g'_{k,i}) = q_k+(2,0)$ for $i>1$, $\gr_N'(h'_{k,i}) = q_k+(1,0)$ for $i<n$, and $\gr_N'(h'_{k,n}) = q_k+(3,0)$. If $q_k$ instead corresponds to a zero of $\Q$ in position $(a,b)$, then we choose ordered bases $(\iota_M(g'_{k,1}),\dots,\iota_M(g'_{k,n}))$ and $(\iota_N(h'_{k,1}),\dots,\iota_N(h'_{k,n}))$ of $M^k$ and $N^k$, respectively, such that
\begin{itemize}
\item[(i)] $\iota_N(h'_{k,1}) = h_{k,b}$,
\item[(ii)] $\iota_M(g'_{k,i}) \in M_{q_k+(1,0)}$ for all $i<n$,
\item[(iii)] $\phi_k^{-1}(\iota_N(h'_{k,i})) = \iota_M(g'_{k,i})$ for all $1\leq i\leq n$,
\end{itemize}
by similar arguments.
Note that the last equation is equivalent to $\phi_k(\iota_M(g'_{k,i})) = \iota_N(h'_{k,i})$.
Let $\gr_N'(h'_{k,1}) = q_k$, $\gr_N'(h'_{k,i}) = q_k+(2,0)$ for $i>1$, $\gr_M'(g'_{k,i}) = q_k+(1,0)$ for $i<n$, and $\gr_M'(g'_{k,n}) = q_k+(3,0)$.

For all $1\leq i\leq n$, let
\begin{align*}
    \iota_M(g'_{\infty,i}) &= g_{\infty,i}, \\\iota_N(h'_{\infty,i}) &= h_{\infty,i},\\
    \gr'_M(g'_{\infty,i}) &= \gr'_N(h'_{\infty,i}) = (C,C).
    \end{align*}
Finally, let $R'_M = \iota_M^{-1}(R_M)$ and $R'_N = \iota_N^{-1}(\phi(R_M))$, and let $\gr'_M(R_M) = \gr'_N(R_N) = (C,C)$.

\begin{lemma}
\label{lem_pres_of_M_and_N}
$\M(P'_M)\cong M$, and $\M(P'_N)\cong N$.
\end{lemma}
\begin{proof}
We will show that $P'_M$ is a regeneration of $P_M$, and that $P'_N$ is a regeneration of $P_N$.
By \cref{lem_reg_iso}, the lemma follows.

By construction, $\iota((G'_M)_{q_k+(i,0)})$ is a basis of $(G_M)_{q_k+(i,0)}$ and $\iota((G'_N)_{q_k+(i,0)})$ is a basis of $(G_N)_{q_k+(i,0)}$ for all $i\in \{0,1,2,3\}$, and $\iota((G'_M)_{(C,C)})$ is a basis of $(G_M)_{(C,C)}$ and $\iota((G'_N)_{(C,C)})$ is a basis of $(G_N)_{(C,C)}$.
We also have $\iota_M(R'_M) = R_M$.
Thus, $P'_M$ is a regeneration of $P_M$.
To show that $P'_N$ is a regeneration of $P_N$, what remains is to show that $\iota_N(\langle R'_N\rangle) = \langle R_N \rangle$.
For $1\leq k\leq K$, let $\rho_{M,k}\colon M^k \to \langle g_{\infty,1}, \dots, g_{\infty,n} \rangle$ be the isomorphism sending $g_{k,i}$ to $g_{\infty,i}$ for all $1\leq i \leq n$.
Define $\rho_{N,k}\colon N^k \to \langle h_{\infty,1}, \dots, h_{\infty,n} \rangle$ similarly.
In the following, we suppress the subscript and write just $\rho$.
We have
\begin{linenomath}\begin{align*}
\langle R_M\rangle &= \langle \{v-\rho(v)\mid 1\leq k\leq K, v\in M^k\}\rangle \text{ and}\\
\langle R_N\rangle &= \langle \{v-\rho(v)\mid 1\leq k\leq K, v\in N^k\}\rangle, \text{ so}\\
\iota_N(\langle R'_N\rangle) &= \langle \iota_N(R'_N)\rangle = \langle \phi(R_M)\rangle = \phi(\langle R_M\rangle)\\
&= \phi(\langle \{v-\rho(v)\mid 1\leq k\leq K, v\in M^k\}\rangle)\\
&= \langle \{\phi(v-\rho(v))\mid 1\leq k\leq K, v\in M^k\}\rangle\\
&= \langle \{\phi(v)-\rho(\phi(v))\mid 1\leq k\leq K, v\in M^k\}\rangle\\
&= \langle \{v-\rho(v)\mid 1\leq k\leq K, v\in N^k\}\rangle\\
&= \langle R_N\rangle,
\end{align*}\end{linenomath}
where we exploit that $\phi$ commutes with both $\langle - \rangle$ and $\rho$.
It follows that $P'_N$ is a regeneration of $P_N$.
\end{proof}

\begin{lemma}
\label{lem_comp_dp}
There is a $\sigma$ such that $P'_M$ and $P'_N$ are $\sigma$-compatible, and $d^p(P'_M,P'_N,\sigma) = (Kn)^{\frac 1p}$.
\end{lemma}
\begin{proof}
Since $\iota_M$, $\phi$ and $\iota_N$ are all isomorphisms, the composition
\[
\sigma\coloneqq \iota_N^{-1}\circ \phi\circ \iota_M\colon \F^{G'_M}\to \F^{G'_N}
\]
is an isomorphism.
By construction, $\sigma$ restricts to bijections $G'_M\to G'_N$ and $R'_M\to R'_N$, so $P'_M$ and $P'_N$ are $\sigma$-compatible.

It remains to compute $d^p(P'_M,P'_N,\sigma)$.
The grades of all the elements of $R'_M$ and $R'_N$ are equal, and for all $g\in G'_M$, $\gr'_M(g)-\gr'_N(\sigma(g)) = (\pm 1,0)$.
To see the latter, note that if for instance $\gr'_M(g) = q_k$, then $g=g'_{k,1}$, so $\sigma(g) = h'_{k,1}$, and $\gr'_N(h'_{k,1}) = q_k+(1,0)$ by (ii) in the construction of $P'_M$ and $P'_N$, since $1<n$.
The other cases are similar or straightforward.
This gives
\begin{linenomath}\begin{align*}
d^p(P'_M,P'_N,\sigma) &= \left(\sum_{g\in G'_M}\|\gr'_M(g)-\gr'_N(\sigma(g))\|^p_p\right)^{\frac 1p}\\
&= \left(\sum_{g\in G'_M\setminus \{g'_{\infty,1},\dots,g'_{\infty,n}\}} 1\right)^{\frac 1p}\\
&= |G'_M\setminus \{g'_{\infty,1},\dots,g'_{\infty,n}\}|^{\frac 1p}\\
&= (Kn)^{\frac 1p}.\qedhere
\end{align*}\end{linenomath}
\end{proof}

\subsubsection{Step 2: If $(\Pp,\Q)$ has no solution, then $d_I^p(M,N)\geq (Kn+1)^{\frac 1p}$.}
\label{subsec_step_2_mods}

We now assume that $(\Pp,\Q)$ has no solution, and will show $d_I^p(M,N)\geq (Kn+1)^{\frac 1p}$.

\begin{definition}
An ordered presentation is a tuple $P=(G,R,\gr,\leq)$ such that
\begin{itemize}
\item $(G,R,\gr)$ is a one-parameter presentation, i.e., a presentation of a one-parameter persistence module,
\item $\leq$ is a total order on $G\cup R$ such that for all $h\leq h'$, we have $\gr(h)\leq \gr(h')$.
\end{itemize}
We define $\M(P)$ as the one-parameter module $\M((G,R,\gr))$.
\end{definition}
For any $q\in \R^2$, let $q_x$ be the $x$-coordinate of $q$.
Given a $2$-parameter presentation $P=(G,R,\gr)$, we define a one-parameter presentation $\bar P=(G,R,\bar \gr)$ by $\bar \gr(h) = \gr(h)_x$ for all $h\in G\cup R$.
If $M=\M(P)$, then we define $\bar M$ as $\M(\bar P)$.
\begin{definition}
Two (two-parameter) presentations
$P = (G,R,\gr)$ and $P' = (G',R',\gr')$ are \emph{strongly compatible} if
\begin{itemize}
\item $G'=G$ and $R'=R$, and
\item there is a total order $\leq$ on $G\cup R$ such that $(G,R,\bar \gr,\leq)$ and $P'=(G,R,\bar \gr',\leq)$ are ordered presentations, and such that we do not have $\bar \gr(h) < z < \bar \gr'(h)$ or $\bar \gr(h) > z > \bar \gr'(h)$ for any $h\in G\cup R$, $z\in \Z$.
\end{itemize}
\end{definition}

\begin{figure}
\centering
\begin{tikzpicture}[scale=1.5]
\node at (-.3,1){$P_M=P_1$};
\draw[shorten <=.3cm, shorten >=.3cm] (0,1) to (.3,0);
\node at (.3,0){$M = M_1\cong M'_1$};
\draw[shorten <=.3cm, shorten >=.3cm] (1.2,1) to (0.9,0);
\node at (1.2,1){$P'_1$};
\draw (1.4,1) to (2.4,1);
\node at (1.9,1.4){strongly};
\node at (1.9,1.14){comp.};
\node at (2.6,1){$P_2$};
\draw[shorten <=.3cm, shorten >=.3cm] (2.6,1) to (2.9,0);
\node at (3.2,0){$M_2\cong M'_2$};
\draw[shorten <=.3cm, shorten >=.3cm] (3.8,1) to (3.5,0);
\node at (3.8,1){$P'_2$};
\node at (4.9,.5){$\dots$};
\begin{scope}[xshift=5cm]
\node at (1,0){$M'_{\ell-1}$};
\draw[shorten <=.3cm, shorten >=.3cm] (1.3,1) to (1,0);
\node at (1.3,1){$P'_{\ell-1}$};
\draw (1.6,1) to (2.5,1);
\node at (2.05,1.4){strongly};
\node at (2.05,1.14){comp.};
\node at (3,1){$P_\ell=P_N$};
\draw[shorten <=.3cm, shorten >=.3cm] (2.7,1) to (3,0);
\node at (3.3,0){$M_\ell = N$};
\end{scope}
\end{tikzpicture}
\caption{The ``zigzag'' of modules and presentations in \cref{lem_refinement_mods}, going from $P_M$ and $M$ on the left to $P_N$ and $N$ on the right through isomorphic modules and strongly compatible presentations.}
\label{fig_M_i_P_i}
\end{figure}

Recall the construction of the presentations $P_M$ and $P_N$ of the modules $M$ and $N$, respectively, in \cref{sec:pmodMN}.
\begin{lemma}
\label{lem_refinement_mods}
For every $\epsilon>0$, there are $\ell\geq 1$ and presentations $P_i = (G_i,R_i,\gr_i)$ for $1\leq i\leq \ell$ and $P'_i = (G_{i+1},R_{i+1},\gr'_i)$ for $1\leq i\leq \ell-1$ such that
\begin{equation}
d_I^p(M,N) + \epsilon > \sum_{i=1}^{\ell-1} d^p(P'_i, P_{i+1}),
\end{equation}
and such that for all $i$,
\begin{itemize}
\item[(i)] $P_1=P_M$ and $P_\ell=P_N$,
\item[(ii)] for $1\leq i<\ell$, $M_i\coloneqq \M(P_i)$ and $M'_i\coloneqq \M(P'_i)$ are isomorphic,
\item[(iii)] for $1\leq i<\ell$, $P'_i$ and $P_{i+1}$ are strongly compatible.
\end{itemize}
\end{lemma}

\begin{proof}
The proof is near identical to the proof of \cref{lem_refinement}. We omit the details; see \cref{fig_M_i_P_i} for an illustration.
\end{proof}

Let $A$ be a module.
For $v\in A^\infty$, define
\[
I(v) = \{p\in \R^2\mid v\in \img A_{p\to \infty}\}.
\]
We shall occasionally write $I^{\gr}(v)$ to make it clear which grade function we are using.

\begin{lemma}
If $f\colon A\to B$ is an isomorphism of modules, then $I(v) = I(f_\infty(v))$ for all $v\in A^\infty$.
\end{lemma}
\begin{proof}
Suppose $s\in I(v)$.
Then there is a $u\in A_s$ such that $A_{s\to \infty}(u) = v$.
We get $B_{s\to \infty}(f_s(u)) = f_\infty(A_{s\to \infty}(u)) = f(v)$, so $s\in I(f(v))$.
It follows that $I(v) \sse I(f_\infty(v))$, and a similar argument using $f^{-1}$ shows that $I(v) \sse I(f_\infty(v))$.
\end{proof}

\begin{lemma}
\label{lem_I_border_gen}
Let $P=(G,R,\gr)$ and $P'=(G,R,\gr')$ be presentations, $v\in \M(P)^\infty = \M(P')^\infty$, and suppose $q\in I^{\gr}(v)$ and $q'\notin I^{\gr'}(v)$.
Then there is a $g\in G$ such that either $\gr(g)_x \leq q_x$ and $\gr'(g)_x > q'_x$, or $\gr(g)_y \leq q_y$ and $\gr'(g)_y > q'_y$.
\end{lemma}

\begin{proof}
We claim that there is a $g\in G$ such that $\gr(g)\leq q$ and $\gr'(g)\nleq q'$.
From this it follows that $\gr(g)_x \leq q_x$ and $\gr(g)_y \leq q_y$, and either $\gr(g)_y > q_y$ or $\gr'(g)_y > q'_y$, which proves the lemma.

We have $\img(\M(P)_{q\to \infty}) = \overline{G^{\gr}_q}$ and $\img(\M(P')_{q'\to \infty}) = \overline{G^{\gr'}_{q'}}$, where $\overline{H}$ is the subspace of $\M(P)_\infty$ generated by $H$.
Since $q\in I^{\gr}(v)$ and $q'\notin I^{\gr'}(v)$, we have $v\in \img(\M(P)_{p\to \infty})\setminus \img(\M(P')_{q'\to \infty})$.
Thus, $G^{\gr}_q\nsubseteq G^{\gr'}_{q'}$, so we can choose a $g\in G^{\gr}_p\setminus G^{\gr'}_{q'}$.
This means that $\gr(g)\leq q$, but $\gr'(g)\nleq q'$, so our claim is proved.
\end{proof}

By \cref{lem_refinement_mods}, we have isomorphisms $M_i\to M'_i$ for all $i$, which induce isomorphisms $\sigma_i\colon (M_i)_\infty\to (M'_i)_\infty = (M_{i+1})_\infty$.
Let $\sigma \coloneqq \sigma_{\ell-1}\circ \dots\circ \sigma_1\colon M^\infty\to N^\infty$.
With respect to the canonical bases, $\sigma$ and $\sigma^{-1}$ have transformation matrices $S$ and $S^{-1}$, respectively.

\begin{lemma}
\label{lem_1_to_3_mods}
There is a $u \in M^\infty$ and a $1\leq k\leq K$ such that either
\begin{itemize}
\item[(i)] $q_k\in I(u)$ and $q_k+(2,2)\notin I(\sigma(u))$, or
\item[(ii)] $q_k\in I(\sigma(u))$ and $q_k+(2,2)\notin I(u)$.
\end{itemize}
\end{lemma}

\begin{proof}
Since $(\Pp,\Q)$ has no solution, either $S$ or $S^{-1}$ has a nonzero element where $\Pp$ or $\Q$, respectively, has a zero.
Suppose $S_{a,b} \neq 0$ and $\Pp_{a,b}=0$ for some $1\leq a,b\leq n$.
Write $v\coloneqq \sigma(\bar g_{\infty,b}) = \sum_{i=i}^n \lambda_i \bar h_{\infty,i}$.
(Recall that $M^\infty$ has a basis $\{\bar g_{\infty,1},\dots,\bar g_{\infty,n}\}$, and $N^\infty$ has a basis $\{\bar h_{\infty,1},\dots,\bar h_{\infty,n}\}$.)
Since the $\bar h_{\infty,i}$ are elements of different direct summands of $N$, we get $I(v) = \bigcap_{\lambda_i \neq 0} I(\bar h_{\infty,i})$.
Since $S_{a,b} \neq 0$, we have $\lambda_a\neq 0$, and thus $I(v)\sse I(\bar h_{\infty,a})$.
By construction, for the $k$ associated to the zero of $\mathcal P$ at $(a,b)$, we have $\gr(\bar g_{k,b}) = q_k$ and $\gr(\bar h_{k,a}) = q_k+(3,0)$ and therefore $q_k\in I(\bar g_{\infty,b})$ and $q_k+(2,0)\notin I(\bar h_{\infty,a})$.
Thus, $q_k+(2,2)\notin I(v)$, so (i) holds.

If $S^{-1}_{a,b} \neq 0$ and $\Q_{a,b}=0$ for some $1\leq a,b\leq n$, then we can prove (ii) using $\sigma^{-1}$ by the same method.
\end{proof}

Fix a $u \in M^\infty$ and a $1\leq k\leq K$ as in \cref{lem_1_to_3_mods}, let $u_i = \sigma_{i-1}\circ \dots\circ \sigma_1(u)\in M_i^\infty$, let $I_i = I(u_i)$, and let $m_i = \min\{c\in \R\mid q_k+(c,c)\in I_i\}$.

In the rest of the proof, we will assume that case (i) holds in \cref{lem_1_to_3_mods} to avoid having to juggle two cases.
Our strategy is to track $m_i$ as it moves from $1$ to $2$ and deduce that this gives a certain cost coming from changes in grades of generators.
If (ii) holds instead, we can run the same argument with $m_i$ moving from $2$ to $1$. The remainder of the proof closely resembles that of the merge tree case. Recall from \cref{sec:pmodMN} that $C=4(Kn+1)^{1/p}$. 

In what follows, we allow barcodes to contain empty intervals of the form $[a,a)$ to avoid having to handle the case $[\gr(g),\gr(r))$ with $\gr(g)=\gr(r)$ separately.
One can check that the cost of matching an interval $[b,c)$ to $[a,a)$ is at least as much as the cost of leaving $[b,c)$ unmatched in a matching of barcodes, so this convention does not allow us to find cheaper matchings.
\begin{lemma}
\label{lem_pairing}
Let $P=(G,R,\gr,\leq)$ be an ordered presentation.
Then there is a $\Pi\sse G\times R$ that does not depend on $\gr$ such that
\begin{itemize}
\item[(i)] each $g\in G$ and $r\in R$ appears in at most one element of $\Pi$, and
\item[(ii)] $B(\M(P)) = \{[\gr(g),\gr(r))\mid (\gr(g),\gr(r))\in \Pi\}\cup \{[\gr(g),\infty)\mid g\notin \Pi_G\}$,
\end{itemize}
where $\Pi_G$ is the set of $g\in G$ appearing in an element of $\Pi$.
\end{lemma}

\begin{proof}
Given $P$ as in the lemma, one can define a matrix $S$ with each row labeled with an element of $G$ and each column labeled with an element of $R$, such that the row (column) labels are increasing with respect to $\leq$ moving downwards (towards the right), and the column of any $r\in R$ is the coordinate vector of $r$ with respect to the ordered basis of $\F^G$ given by the row labels.
The standard persistence algorithm \cite[Ch.~VII.1]{edelsbrunner2022computational} outputs a pairing of rows with columns with each row and column appearing in at most one pair, and the barcode can be read off as in (ii).
Moreover, this pairing of rows with columns is independent of $\gr$, since by the pairing lemma, it depends only on the (unlabeled) matrix $S$.
\end{proof}
Recall that for a presentation $(G,R,\gr)$ and $h\in G\cup R$, $\bar \gr(h)$ is the $x$-coordinate of $\gr(h)$.
\begin{lemma}
\label{lem_subdivide}
Suppose $M_1, M_2, \dots, M_\ell$ and $P'_1,P_2,P'_2,\dots, P'_{\ell-1},P_\ell$ satisfy the conditions of \cref{lem_refinement_mods}.
Then, for every $1\leq i<\ell$, there is a subset $\Pi_i \subseteq G_i\times R_i$, with each $g\in G_i$ and $r\in R_i$ showing up in at most one pair of $\Pi_i$, such that
\begin{align*}
B(\bar M_i) &= \{[\bar \gr'_i(g),\bar \gr'_i(r))\mid (g,r)\in \Pi_i\}\cup \{[\bar \gr'_i(g),\infty)\mid g\notin (\Pi_i)_{G_i}\},\\
B(\bar M_{i+1}) &= \{[\bar \gr_{i+1}(g),\bar \gr_{i+1}(r))\mid (g,r)\in \Pi_i\}\cup \{[\bar \gr_{i+1}(g),\infty)\mid g\notin (\Pi_i)_{G_i}\},
\end{align*}
where $(\Pi_i)_{G_i}$ is the set of $g\in G_i$ appearing in an element of $\Pi_i$.
\end{lemma}

\begin{proof}
By \cref{lem_refinement_mods}, for $1\leq i<\ell$, $B(\bar M_i)=B(\bar M'_i)$, and $P'_i = (G_{i+1},R_{i+1},\gr'_i)$ and $P_{i+1} = (G_{i+1},R_{i+1},\gr_{i+1})$ are strongly compatible.
Thus, we can find a total order $\leq_{i+1}$ on $G_{i+1}\cup R_{i+1}$ such that $(G_{i+1},R_{i+1},\bar \gr'_i,\leq_{i+1})$ and $(G_{i+1},R_{i+1},\bar \gr_{i+1},\leq_{i+1})$ are ordered presentations.
Applying \cref{lem_pairing} to the two ordered presentations, we get a $\Pi_i$ satisfying the conditions of the lemma.
\end{proof}

\begin{corollary}
\label{cor_matching_mods} Suppose $M_1, M_2, \dots, M_\ell$ and $P'_1,P_2,P'_2,\dots, P'_{\ell-1},P_\ell$ satisfy the conditions of \cref{lem_refinement_mods}.
Then $\Pi_i$ defines a matching $\phi_i\colon B(\bar M_i)\to B(\bar M_{i+1})$ such that
\begin{align*}
\phi_i([\bar \gr_i'(g), \bar \gr_i'(r))) &= [\bar \gr_{i+1}(g), \bar \gr_{i+1}(r)), \qquad (g,r)\in \Pi_i\\
\phi_i([\bar \gr_i'(g), \infty)) &= [\bar \gr_{i+1}(g),\infty), g\notin (\Pi_i)_{G_i},
\end{align*}
where $(\Pi_i)_{G_i}$ is the set of $g\in G_i$ appearing in an element of $\Pi_i$.
\end{corollary}
\begin{proof}
Immediate.
\end{proof}

The barcode of $\bar M$ has one interval for each generator in $G_M$, and each interval is of the form $[a,b)$, where $a$ is an even integer and $b$ is either $C$ or $\infty$.
Similarly, the barcode of $\bar N$ has one interval for each generator in $G_N$, and each interval is of the form $[a,b)$, where $a$ is an odd integer and $b$ is either $C$ or $\infty$.
Moreover, $|G_M|=|G_N|=Kn$.
Note that each interval of $B(\bar M)$ has length at least $C=(Kn+1)^{\frac 1p}$.

\begin{lemma}
\label{lem_bij_mods}
Suppose $d_I^p(M,N) < (Kn+1)^{\frac 1p}$.
For $0 < \epsilon < (Kn+1)^{\frac 1p}$ in \cref{lem_refinement_mods}, the composition of matchings 
\[ \phi = \phi_{l-1}\circ \cdots \circ  \phi_1 \colon B(\bar M)\to B(\bar N)\]
is a bijection. 
\end{lemma}

\begin{proof}
Suppose that an interval $I\in B(\bar M)$ is left unmatched by $\phi$.
Every interval in $B(\bar M)$ has length at least $C$, so $I$ contributes at least $\frac C2$ to $p\text{-cost}(\phi)$.
Using this, \cref{cor_matching_mods}, and that the $p$-cost of a composition of morphisms is at most the sum of the $p$-costs of each morphism, we get
\[
\frac C2 \leq \text{$p$-cost}(\phi) \leq \sum_{i=1}^{\ell-1} \text{$p$-cost}(\phi_i) \leq  \sum_{i=1}^{\ell-1} d^p(P_i', P_{i+1}) < d_I^p(M,N)+\epsilon < 2(Kn+1)^{\frac 1p},
\]
but this contradicts $C=4(Kn+1)^{\frac 1p}$.
\end{proof}

In order to conclude the proof, we must introduce some notation.
First, order the intervals of $B(\bar M)$ as $B(\bar M) = \{I_1^1, I_2^1, \ldots, I_{Kn}^1\}$.
For $2\leq i\leq \ell$, and $1\leq j\leq Kn$, let $I_j^i = \phi_{i-1}\circ \cdots \circ \phi_1(I_j^1)\in  B(\bar M_{i+1})$, where $\phi_i$ is as in \cref{cor_matching_mods}.
By \cref{cor_matching_mods}, we can pick $g_j^i$ and $r_j^i$ with $g_j^i\neq g_{j'}^i$ for $j\neq j'$ such that $I_j^i = [\bar \gr_i'(g_j^i), \bar \gr_i'(r_j^i))$ and $I_j^{i+1} = [\bar \gr_{i+1}(g_j^i),\bar \gr_{i+1}(r_j^i))$.
Furthermore, let $v_i = (v_1^i, \ldots, v_{n-1}^i)$, where $v_{j}^i=\bar \gr_i'(g_j^i)$ is the left endpoint of $I_j^i$.
Let $\hat{v}_j^i = \min|v_j^i-z|$, where the minimum is taken over all odd integers $z$.
Finally, let $\delta v_j^i = \max(\hat{v}_{j}^{i} - \hat{v}_j^{i+1},0)$.
In other words, it is only positive if the left endpoint is moving away from an even number, towards an odd number.

Next, let $\delta m_i = m_{i+1}-m_i$ if $1\leq m_i\leq m_{i+1}\leq 2$, and $\delta m_i = 0$, otherwise.

\begin{lemma}
For all $1\leq i\leq l$, 
\[
d^p(P_i',P_{i+1}) \geq ||(\delta v_1^i, \ldots, \delta v_{Kn}^i, \delta m_i)||_p.
\]
\label{lem_single-step}
\end{lemma}

\begin{proof}
We have
\begin{equation}
\label{eq_dp}
d^p(P_i',P_{i+1})\geq \left(\sum_{g\in G_{i+1}} |\gr_{i+1}(g)_x - \gr_i'(g)_x|^p + |\gr_{i+1}(g)_y - \gr_i'(g)_y|^p\right)^{1/p},
\end{equation}
so taking the $p$th power of both sides of the inequality in the lemma, we only need to show
\begin{equation}
\sum_{g\in G_{i+1}} \left(|\gr_{i+1}(g)_x - \gr_i'(g)_x|^p + |\gr_{i+1}(g)_y - \gr_i'(g)_y|^p\right)^{1/p} \geq (\delta m_i)^p + \sum_{j=1}^{Kn} (\delta v_j^i)^p.
\end{equation}

By definition, we have $\delta v_j^i \leq |\gr_{i+1}(g_j^i) - \gr_i'(g_j^i)|$ for all $i$ and $j$.
Let $q'=q_k+(m_i,m_i)$ and $q=q_k+(m_{i+1},m_{i+1})$.
By \cref{lem_I_border_gen}, there is an $h\in G_{i+1}$ such that either $\gr'_i(h)_x \leq q'_x$ and $\gr_{i+1}(h)_x \geq q_x$, or $\gr'_i(h)_y \leq q'_y$ and $\gr_{i+1}(h)_y \geq q_y$.
Thus, we have
\[
\delta m_i\leq |\gr_{i+1}(h)_x- \gr'_i(h)_x|
\]
in the first case and
\[
\delta m_i\leq |\gr_{i+1}(h)_y- \gr'_i(h)_y|
\]
in the second.

Thus, every term on the right hand side of \cref{eq_dp} is bounded above by a term on the left hand side, so to prove the lemma, it suffices to show that we have not used a term on the left hand side to bound two strictly positive terms on the right hand side.
This can only happen if $h=g_j^i$ for some $j$, we are using $\delta m_i\leq |\gr_{i+1}(h)_x- \gr'_i(h)_x|$, and $\delta v_j^i$ and $\delta m_i$ are both strictly positive.

In this case, we have $\gr'_i(h)_x\leq (q_k)_x + 2$ and $\gr_{i+1}(h)_x \geq (q_k)_x + 1$, so since $P_i'$ and $P_{i+1}$ are strongly compatible by \cref{lem_refinement_mods} (iii) and $\delta m_i>0$, we have
\[
(q_k)_x + 1\leq \gr'_i(h)_x<\gr_{i+1}(h)_x\leq (q_k)_x + 2.
\]
By construction, $(q_k)_x$ is even, so since $h=g_j^i$, we have $\hat{v}_j^{i} - \hat{v}_j^{i+1}<0$, and thus, $\delta v_j^i = 0$, which is a contradiction.
\end{proof}

\begin{theorem}
If $(\Pp,\Q)$ has no solution, then $d_I^p(M,N) \geq (Kn+1)^{1/p}$. 
\end{theorem}
\begin{proof}
Let $d_I^p(M,N) > \epsilon>0$.
We use the notation introduced after \cref{lem_bij_mods}.
Let $\Delta(I^1_j)= \sum_{i=1}^{\ell-1} \delta v_j^i$ and $\Delta(m) = \sum_{i=1}^{\ell-1} \delta m_i$.
Since all the left endpoints of intervals in $B(\bar M)$ are even integers and the left endpoints of intervals in $B(\bar N)$ are odd integers, we have $\hat{v}_1^i = 1$ and $\hat{v}_\ell^i = 0$, and thus $\Delta(I^1_j)\geq 1$.
Moreover, $\Delta(m)\geq 1$ because $m_1 \leq 0$ and $m_\ell > 2$ by \cref{lem_1_to_3_mods} (i), and the $m_i$ must therefore traverse the interval from $1$ to $2$.
By \cref{lem_refinement_mods} and \cref{lem_single-step},
\begin{align*}
d_I^p(M,N) + \epsilon > \sum_{i=1}^{\ell-1} d^p(P_i', P_{i+1}) &\geq \sum_{i=1}^{\ell-1} ||(\delta v_1^i, \ldots, \delta v_{Kn}^i, \delta m_i)||_p\\
\text{(triangle inequality)}&\geq ||(\Delta(I_1^1), \ldots, \Delta(I_{Kn}^1), \Delta(g,h))||_p \\
& \geq ||(1,1, \ldots,1)||_p \\
& = (Kn+1)^{1/p}.
\end{align*}
Since this holds for every sufficiently small $\epsilon>0$, the lemma follows.
\end{proof}

\textbf{Data availability statement:} No data was produced or analyzed.

\section{Discussion and future work}
A natural question to consider is the problem of approximating $p$-presentation distances within a constant factor $c$; for $p=\infty$ this is known to be NP-hard for $c<3$ both for merge trees \cite{agarwal2018computing} (see Theorem 3.3 and footnote 3) and for multiparameter modules \cite[Theorem 4]{bjerkevik2020computing}.
The techniques used to prove those results do not generalize immediately to the case of $p<\infty$, as the proofs implicitly involve cost vectors of the form $u=(1,\dots,1)$ and $v=(1,\dots,1,3)$ for which the ratio $\|v\|_\infty/\|u\|_\infty = 3$ gives the factor of 3. In contrast, $\|v\|_p/\|u\|_p$ goes to $1$ as the size of the vectors increases, and therefore we cannot get hardness of approximation up to any constant factor for finite $p$ using the methods of \cite{agarwal2018computing,bjerkevik2020computing}. 
It is not unreasonable to hope for positive approximation results; for instance, it is worth exploring whether polytime approximation algorithms for partition problems can be used to give a polytime algorithm for approximating $d_I^p$ for merge trees.

It remains an open question whether computing the $p$-presentation distance is in NP for $p < \infty$, both for merge trees and persistence modules. The main challenge is that these distances are defined as an infimum over sequences of merge trees or persistence modules, and currently there is no known upper bound on the length $\ell$ of such sequences.
A closely related open problem is whether the infimum can be replaced by a minimum.

\bibliographystyle{plainurl}
\bibliography{ref}
\end{document}